\documentclass[runningheads]{llncs}
\usepackage[T1]{fontenc}
\usepackage{graphicx}
\usepackage{graphicx}
\usepackage{enumerate}
\usepackage{xcolor}
\usepackage{comment}
\usepackage{tabularx}
\usepackage{enumitem}
\usepackage{orcidlink}
\usepackage{mathtools}
\usepackage{mathrsfs}
\usepackage{todonotes}
\usepackage{caption}
\usepackage{subcaption}
\usepackage{comment}
\usepackage{todonotes}
\usepackage{complexity}
\usepackage{amssymb}
\usepackage{xargs} 
\usepackage{algorithm}
\usepackage{algorithmic}
\definecolor{MyBlue}{RGB}{50,100,200}
\usepackage{booktabs}
\usepackage{multirow}
\usepackage{hyperref}
\hypersetup{%
  colorlinks=true,
  linkcolor=MyBlue,
  citecolor=MyBlue,
  urlcolor=MyBlue
}
\usepackage[nameinlink,capitalise]{cleveref}

\usepackage{svg}
\newtheorem{assumption}{Assumption}
 \crefname{assumption}{Assumption}{Assumptions}

\usepackage{pgfplots}
\pgfplotsset{compat=1.18}
\usetikzlibrary{patterns}
\usepackage{thm-restate}

\newcommand{\ideal}{\textsc{Ideal}}

\newcommand{\isbad}{\textsf{flag}}
\newcommand{\mupdate}{\textsf{nextstate}}

\newcommand{\Nb}{\mathbb{N}}
\newcommand{\Gb}{\mathbb{G}}
\newcommand{\Zb}{\mathbb{Z}}

\newcommand{\Cc}{\mathcal{C}}
\newcommand{\Dc}{\mathcal{D}}

\newcommand{\Gc}{\mathcal{G}}
\newcommand{\Hc}{\mathcal{H}}

\newcommand{\Sc}{\mathcal{S}}

\newcommand{\Pc}{\textsf{Sys}}
\newcommand{\Mc}{\textsf{Mon}}

\newcommand{\seq}[1]{\left\langle #1 \right\rangle}
\newcommand{\tpl}[1]{\left( #1 \right)}

\newcommand{\set}[1]{\left\{ #1 \right\}}

\newcommand{\setup}{\texttt{Setup}}

\newcommand{\Oh}{\mathcal{O}}

\newcommand{\distribution}{\Hc}

\newcommand{\Enc}{\text{Enc}}

\newcommand{\encYao}{\text{encYao}}

 \newcommand{\protocolone}{Open Specification Protocol}
 \newcommand{\protocoltwo}{Hidden Specification Protocol}

\newcommand{\drawlock}[2]{ 
    \begin{scope}[shift={(#1,#2)}, scale=0.01]  
        \draw[thick] (-10,20) arc[start angle=180, end angle=0, radius=10] -- (10,5) -- (7,5) 
                     -- (7,15) arc[start angle=0, end angle=180, radius=7] -- (-7,5) -- (-10,5) -- cycle;

        \draw[thick, fill=gray!30] (-15,5) rectangle (15,-20);

        \draw[fill=black] (0,-10) circle (3);
        \draw[thick] (-2,-10) -- (2,-10) -- (0,-16) -- cycle;
    \end{scope}
}

\newcommand\indistinguishable{\mathrel{\stackrel{\makebox[0pt]{\mbox{\normalfont\tiny c}}}{\equiv}}}

\newcommand{\systemval}[2]{\ifthenelse{\equal{#2}{}}{\texttt{pval}_{#1}}{\texttt{pval}_{#1}\left({#2}\right)}}
\newcommand{\monitorval}[2]{\ifthenelse{\equal{#2}{}}{\texttt{rval}_{#1}}{\texttt{sval}_{#1}\left({#2}\right)}}

\newclass{\TOWER}{TOWER}
\newclass{\ACKERMANN}{ACKERMANN}
\newclass{\EXPTIME}{EXPTIME}
\newclass{\IIEXPTIME}{2-EXPTIME}
\newclass{\NEXPTIME}{NEXPTIME}

\renewcommand{\leq}{\leqslant}
\renewcommand{\geq}{\geqslant}

\newcommand{\encGG}{\textbf{encGG}}
\newcommand{\fakeGarble}{\overline{\textbf{Garble}}}
\newcommand{\Garble}{\textbf{Garble}}
\newcommand{\fakeKeys}{\overline{\textbf{Keys}}}
\newcommand{\keys}{\textbf{Keys}}

\newcommand{\view}{\texttt{view}}

\newcommand{\proceed}{\textsc{proceed}}
\newcommand{\terminate}{\textsc{terminate}}
\newcommand{\pstring}{\sigma}
\newcommand{\mstring}{\mu}
\newcommand{\OT}{\texttt{OT}}

\newcommandx{\theju}[2][1=]{\todo[linecolor=blue,backgroundcolor=blue!25,bordercolor=blue,#1]{\tiny T: #2}}
\newcommandx{\mahyar}[2][1=]{\todo[linecolor=red,backgroundcolor=red!25,bordercolor=red,#1]{\tiny M:#2}}

\begin{document}
\title{Privacy-Preserving Runtime Verification}
%
%
\author{
Thomas A. Henzinger\inst{1}\orcidlink{0000-0001-6077-7514} 
Mahyar Karimi\inst{1}\orcidlink{0009-0005-0820-1696}\and
K. S. Thejaswini\inst{1}\orcidlink{0000-0001-6077-7514}}
\authorrunning{T. A. Henzinger, M. Karimi, and K. S. Thejaswini}
%
\institute{Institute of Science and Technology Austria, Klosterneuberg, Austria
 \email{\{tah,mahyar.karimi,thejaswini.k.s\}@ista.ac.at}}
\maketitle              
\begin{abstract}
Runtime verification offers scalable solutions to improve the safety and reliability of systems. However, systems that require verification or monitoring by a third party to ensure compliance with a specification might contain sensitive information, causing privacy concerns when usual runtime verification approaches are used. Privacy is compromised if protected information about the system, or sensitive data that is processed by the system, is revealed. In addition, revealing the specification being monitored may undermine the essence of third-party verification. 

In this work, we propose two novel protocols for the privacy-preserving runtime verification of systems against formal sequential specifications. In our first protocol, the monitor verifies whether the system satisfies the specification without learning anything else, though both parties are aware of the specification. Our second protocol ensures that the system remains oblivious to the monitored specification, while the monitor learns only whether the system satisfies the specification and nothing more. Our protocols adapt and improve existing techniques used in cryptography, and more specifically, multi-party computation.

The sequential specification defines the observation step of the monitor, whose granularity depends on the situation (e.g., banks may be monitored on a daily basis). Our protocols exchange a single message per observation step, after an initialisation phase. This design minimises communication overhead, enabling relatively lightweight privacy-preserving monitoring. We implement our approach for monitoring specifications described by register automata and evaluate it experimentally.


\end{abstract}
\section{Introduction}\label{sec:intro}
Recent advances have demonstrated that verification can be performed in a privacy-preserving manner spanning a variety of domains, including SAT solving~\cite{LJAPW22}, verifying resolution proofs~\cite{LAHPTW22}, matching strings against regular expressions~\cite{LWSTRP24}, and also model-checking specifications described by CTL formulas~\cite{JLAP20}.
However, privacy-preserving verification often introduces significant computational overhead. This poses scalability challenges, particularly in real-world applications where the underlying systems involve an enormous number of states. 
In contrast, runtime verification~\cite{LS09,BFFR18} focuses not on verifying the entire system but rather on monitoring specific outputs produced by the system during execution. This approach inherently involves smaller-scale data exchanges, making it a suitable candidate for privacy-preserving methods. 


Privacy and monitoring seem at odds. Monitoring typically involves verifying whether the execution trace of a system satisfies a given specification, which seemingly requires access to the system trace. Privacy, on the other hand, demands that protected information about the system remain undisclosed. Balancing these conflicting objectives presents a significant challenge, particularly in contexts where monitoring is essential but the underlying data is sensitive. This challenge arises in many real-world scenarios; a bank undergoing an audit may need to share data to demonstrate compliance, a hospital may need to report statistics to ensure adherence to healthcare standards, a self-driving car or metro system may be monitored for its operational safety. In each case, sharing protected IP or client data with a third-party monitor raises concerns about privacy.
One possible solution is to allow the system to self-monitor by embedding the specification into the system and generating its own verification results. However, this approach lacks credibility in settings where third-party validation is crucial. For instance, if hospitals were entirely self-certified, the certificates would lack the impartiality needed to inspire public confidence. 
This necessitates a solution that allows third-party monitors to continuously and repeatedly verify whether a system complies with a given specification, all while obscuring the internals of the system and its data.


Hiding the system and its data is just a first step, but our need for privacy might not stop there. Sometimes, one might even need to hide the specification that a system is being monitored against. 
For motivation, we can again consider the domain of healthcare, where the runtime verification of traces satisfying specifications expressed in temporal logic has already been studied~\cite{JLKWHSS16}. 
For instance, in the National Health Services (NHS) in England, hospital funding is tied to performance metrics, optimisation of which led many hospitals to restructure their operations~\cite{Cra17,Mea14}. Unfortunately, such restructurings resulted in extreme cases like the events at Mid-Staffordshire NHS Foundation Trust, where financial targets led to patient neglect~\cite{mid2013report}. This exemplifies Goodhart's Law: \emph{``When a measure becomes a target, it ceases to be a good measure,''} or, as an ex-NHS manager put it, \emph{``hitting the target but missing the point.''} Indeed any measure has a potential to become a target when gamifications are employed to achieve that target. 
To avoid such distortions, it is crucial to also develop protocols for monitoring that keep both the system/data and specification private, ensuring accurate and unbiased evaluation.

\paragraph{Cryptography.} We use tools developed in cryptography to accomplish our privacy-preserving algorithms in the settings described above. 
 Our privacy-preserving monitoring algorithms rely on techniques from multi-party computation (MPC), studied since the 1980s~\cite{Yao82}. MPC allows two or more parties, who do not trust each other, to collaboratively compute a function on their private inputs without revealing these inputs to each other. Typically, MPC focuses on a ``one-shot'' setting, where the parties compute one or more outputs based on their inputs, but do so exactly once. 
We use a specific variant of MPC called private function evaluation (PFE), where one party owns a private function (a ``secret circuit'') along with a part of the input, while the other party holds the rest of the input to that function. The goal is to design a protocol that allows the function to be evaluated using the inputs of both parties securely: one or both parties can learn the result of the function on the input, but nothing else is revealed. Essentially, PFE enables computation as if the function owner shared their circuit and the other shared their input, without actually doing so.

Runtime verification, however, requires not a one-shot but a repeated evaluation process, whose granularity---called ({\em observation}) {\em round}---depends on the application.
The performance of banks may be observed daily, hospitals monthly, autopilots and metro systems every second. 
Although most existing PFE protocols are designed for single-use computations, a few have been modified for efficient repeated use~\cite{MS13,BBKL19,LWY22};
however, these usually do not account for situations where the computation needs to be repeated many times, nor for maintaining internal states that are kept secret from all parties. 
Such {\em reactive functionalities} are important for monitoring systems, where internal states of the system (e.g., accumulated data or ongoing logs) and the monitor (for sequential specifications) need to remain hidden even across repeated interactions. 
Protocols that can be adapted to reactive functionalities using standard cryptographic techniques like secret sharing require at least as many message exchanges per round as Oblivious Transfer (cryptographic protocols in which a receiver is to obtain one of many messages from a sender without revealing their choice), which exchanges three messages~\cite{CSW20} and causes significant computational overhead.

\paragraph{Our contribution.}
We provide protocols that would aid a monitoring setup as show in \cref{fig:monitoringsetup}. Implementing our protocols on the system's side, and monitor's side enables monitoring where only one message is sent per observation round, the specification is kept secret, and the observable outputs of the system are kept secret. 
Instead of using secret sharing to construct PFE protocols for reactive functionalities, we propose novel protocols designed specifically for monitoring safety properties. 
Our protocols send a single message per round from the system to the monitor, reducing computational cost while still ensuring privacy.
The specification is given as a state machine with a next-state function that, in each round, updates the specification state based on an observation of the system.
Additionally, the specification maintains a boolean output called \emph{flag}.
Our protocols ensure that the monitor learns about the monitored system only a single bit per round---the flag---which indicates whether or not the specification is satisfied for the prefix of the trace observed so far. 
In this way, the monitor can be convinced of the trace's correctness without knowing the input, output, nor internal 
state of the monitored system in any round, nor even the internal state of the specification itself.

We consider two security settings for which we provide privacy-preserving protocols for safety monitoring:
\begin{itemize}
    \item[(a)] [Open specification] The specification is not a secret and known to both parties.
    \item[(b)] [Hidden specification] The specification is a secret and known only to the monitor. 
\end{itemize}
Note that in setting~(a), the monitored system knows which of its parts (which system inputs, which system outputs, and which internal system states) are being looked at by the specification, 
while in setting~(b) it does not.
Hence setting (b) is particularly interesting for large systems being monitored against small specifications. 

Our privacy-preserving protocols are obtained by first converting the specification into a boolean circuit that computes a next-state function along with the flag function. 
For setting~(a), where the specification is not a secret, we produce a protocol that is a  modification of the classic MPC protocol known as Yao's garbled circuits~\cite{Yao86,GMW87}. 
The intuition behind garbled circuits is that the monitored system encrypts each gate of the specification circuit and sends this encrypted value to the monitor. 
We modify the classical protocol to enable repeated computation while maintaining the secrecy of the specification state.  
For setting~(b), where only the monitor knows the specification, 
inspired by recent advances in PFE, we provide novel protocols that compute reactive functionalities with low computational overhead. 
We build upon the seminal work of Katz and Malka~\cite{KM11} and the recent work of Liu, Wang, and Yu~\cite{LWY22}, which has built on other PFE-related works~\cite{MS13,BBKL19}.
Our protocols are designed so that the system sends only one message per round to the monitor.
Although there is an initial setup phase that may involve multiple message exchanges, this is a one-time cost and does not affect the ongoing performance of the protocol during runtime. This is also helpful in monitoring situations where the hardware for bidirectional communication is unavailable.

We implemented our protocols to analyse the influence of several parameters involved in building such protocols. 
We use specifications described as register automata~\cite{GDPT13}, which we convert to boolean circuits. 
Our experiments show the feasibility of our protocols when the circuit sizes are on the order of $10^5$ for acceptable security parameters. 
Additionally, in our second protocol designed for hidden specifications, the time per round is influenced more by the size of the specification than by the size of the monitored system.
This allows scalability to large system sizes as long as the specification remains small. 
Indeed, our protocol is significantly more scalable than those in recent related works~\cite{BMMBWS22,WMSMBBS24} which, although similar in motivation, address fundamentally different problems in distinct settings. 

\paragraph{Related works.} 
Recent work on privacy-preserving verification has largely focused on static settings. A wide range of verification paradigms have also been shown to be amenable to cryptographic techniques: for example, verifying resolution proofs in zero knowledge~\cite{LAHPTW22}, solving SAT formulas~\cite{LJAPW22}, matching strings against regular expressions~\cite{LWSTRP24}, checking string and regular expression equivalence~\cite{KAAP25}, and model-checking CTL specifications~\cite{JLAP20}. In contrast, our work targets runtime verification, where the system’s data evolves dynamically and must be processed incrementally at each step.

With a similar motivation of monitoring specifications in a privacy preserving manner, Banno et al.~\cite{BMMBWS22}, and later Waga et al.~\cite{WMSMBBS24} 
provided algorithms for oblivious online monitoring for Linear Temporal Logic (LTL) specifications~\cite{Pnu77} and Signal Temporal Logic (STL)~\cite{MN04} specifications, respectively using Fully Homomorphic Encryption (FHE).
Although at first glance Banno et al. and Waga et al. may appear to solve the same problem, the settings for ours and their problems are different. 
Their main objective is to ensure monitoring of specifications where the computation can be outsourced to a second party, or an external server. 
Further, in their setting, this server knows the specification, however learns neither the system's observable outputs, nor if the specification is satisfied. Only the system itself learns if the specification is satisfied. 

Due to our focus on enabling privacy-preserving third-party monitoring, in our setting, it is imperative that the party that knows the specification (monitor) can also learn if this specification is being satisfied, but not the system's observable outputs. Their protocol would therefore not be an appropriate solution in our setting, since it cannot be modified to make the second party that knows the specification learn whether the specification is satisfied, without also learning the system's observable outputs. 
On the other hand, both the related works deal with malicious Systems---a setting our protocol does not extend to, and hence our protocols would not be a appropriate solution in their setting either. 

In terms of specifications, their work handles specifications represented using temporal logic that is later converted into finite state automata. We highlight that we consider circuits to described our specifications, and this representation is general enough that it handles all sequential specifications, including LTL, STL, or finite state automata, much more succinctly.  


Since our work and their work deal with different settings, it is not ideal to compare either work with each other directly. However, since certain specifications considered by them can be of interest to us and vice-versa, we run our protocols on their specifications and extrapolate the time taken by theirs on ours in \cref{sec:experiments}, where we show that our protocol handles all specs in the work of Bano et al.~\cite{BMMBWS22} within a few hundred milliseconds. 


Our protocols for privacy-preserving runtime verification draw inspiration from several works on two-party computation~\cite{GMW87,Yao82}, and more specifically private function evaluation~\cite{KM11,MS13,BBKL19,LWY22}. Our Hidden Specification Protocol draws specifically on the recent work of Liu, Wang, and Yu~\cite{LWY22}. In their work they provide a way to repeatedly compute the same function several times under standard IND-CPA and DDH assumptions against \emph{covert adversaries}, a more adversarial setting than ours. However, Liu et al's protocol does not allow for hiding an internal monitor state (reactive functionalities).
\begin{figure*}[h]
\centering
\begin{tikzpicture}[scale = 1, transform shape]
    \tikzset{
        pics/fbloop/.style n args={4}{  
          code={
          \def\width{#1}   
          \def\x{#2}       
          \def\y{#3}       
          \def\word{#4}
          \node[draw, rectangle, minimum width=\width cm, minimum height=\width cm] (box) at (\x,\y) {\word};
          \draw[->]
            (\x+0.5*\width,\y+0.2*\width)
            -- ++(0.3*\width,0)     
            -- ++(0,0.6*\width)  
            -- ++(-1.6*\width,0)    
            -- ++(0,-0.6*\width) 
            -- ++(0.3*\width,0);    
           \draw[->] (\x+0.5*\width,\y-0.2*\width) --  (\x+\width,\y-0.2*\width);
           \draw[->]   (\x-\width,\y-0.2*\width) -- (\x-0.5*\width,\y-0.2*\width);
          }
        }  
    }

    \tikzset{
        pics/fbloopnoout/.style n args={4}{  
          code={
          \def\width{#1}   
          \def\x{#2}       
          \def\y{#3}       
          \def\word{#4}
          \node[draw, rectangle, minimum width=\width cm, minimum height=\width cm] (box) at (\x,\y) {\word};
          \draw[->]
            (\x+0.5*\width,\y+0.2*\width)
            -- ++(0.3*\width,0)     
            -- ++(0,0.6*\width)  
            -- ++(-1.6*\width,0)    
            -- ++(0,-0.6*\width) 
            -- ++(0.3*\width,0);    
           \draw[->]   (\x-\width,\y-0.2*\width) -- (\x-0.5*\width,\y-0.2*\width);
          }
        }  
    }

    \tikzset{
        pics/fbloopnoin/.style n args={4}{  
          code={
          \def\width{#1}   
          \def\x{#2}       
          \def\y{#3}       
          \def\word{#4}
          \node[draw, rectangle, minimum width=\width cm, minimum height=\width cm] (box) at (\x,\y) {\word};
          \draw[->]
            (\x+0.5*\width,\y+0.2*\width)
            -- ++(0.3*\width,0)     
            -- ++(0,0.6*\width)  
            -- ++(-1.6*\width,0)    
            -- ++(0,-0.6*\width) 
            -- ++(0.3*\width,0);    
            \draw[->] (\x+0.5*\width,\y-0.2*\width) --  (\x+\width,\y-0.2*\width);
          }
        }  
    }

    \tikzset{
        pics/fbloopsystem/.style n args={4}{  
          code={
          \def\width{#1}   
          \def\x{#2}       
          \def\y{#3}       
          \def\word{#4}
          \node[fill=black!70,draw, rectangle, minimum width=\width cm, minimum height=\width cm] (box) at (\x,\y) {\word};
          \draw[->]
            (\x+0.5*\width,\y+0.2*\width)
            -- ++(0.3*\width,0)     
            -- ++(0,0.6*\width)  
            -- ++(-1.6*\width,0)    
            -- ++(0,-0.6*\width) 
            -- ++(0.3*\width,0);    
           \draw[->]   (\x-\width,\y-0.2*\width) -- (\x-0.5*\width,\y-0.2*\width);
          }
        }  
    }
        

    \draw[rounded corners=5pt, thick, fill=gray!20] (-0.9,2.7) rectangle ++ (1.8,0.9);
    \drawlock{-0.5}{3.1};
    \pic {fbloop={0.4}{0.3}{3.1}{$\Phi$}};
    \node at (0.3,3.5) {\tiny nextstate};
    \node at (0.7,2.85) {\scalebox{0.7}{\tiny flag}};

    \node at (0,2.4) {\tiny Round 0: Setup};

    \draw[thick, rounded corners=5pt] (-6.9,3) rectangle (-1.1,0);
    \node at (-3.5,-0.3) {System};
    \draw[dashed, rounded corners=5pt] (-6.7,2.7) rectangle (-3.7,0.3) node[at={(-5.2,0.2)}] {\tiny Unmonitored system};
    \draw[dashed, rounded corners=5pt, fill=gray!20] (-3.5,2.7) rectangle (-1.2,0.3) node[at={(-2.3,0.2)}] {\tiny Protocol: System side};

    \pic {fbloopsystem={1.2}{-5.2}{1.3}{}};
    \node[white] at (-5.2,1.6) {\scriptsize Unknown};
    \node[white] at (-5.2, 1.4) {\scriptsize system};
    \node[white] at (-5.2, 1.2) {\scriptsize state-};
    \node[white] at (-5.2, 1) {\scriptsize machine};
    
    \pic {fbloopnoin={0.9}{-2.4}{1.2}{$\pi_\Pc$}};

    \draw[->] (-4.6,1) -- (-2.85,1);
    \draw[fill=white, white] (-4.3,0.4) rectangle ++ (0.82,0.5);
    \node at (-3.7,0.8) {\tiny Observable};
    \node at (-3.7,0.6) {\tiny output};
    
    \draw[rounded corners=5pt,thick] (1,0) rectangle (5,3);
   \node at (3.5,-0.3) {Safety monitor};

    \draw[dashed, rounded corners=5pt] (4.8,2.7) rectangle (3.4,1.6) node[at={(4.2,1.5)}] {\tiny Monitor Spec.};
    \draw[dashed, rounded corners=5pt,fill=gray!20] (3.2,2.7) rectangle (1.2,0.3) node[at={(2.3,0.2)}] {\tiny Protocol: Monitor side};
    
    \draw[->, thick] (3.5,2.3) 
                                -- ++ (-0.9,0)
                                -- ++ (0,-0.2);

    \draw[->, thick, double] (1.8,2.2) -- ++ (0,0.4)
                               -- ++ (-4,0)
                               -- ++ (0,-0.4);    

    \pic {fbloopnoout={0.9}{2.4}{1.2}{$\pi_\Mc$}};
    \pic {fbloop={0.4}{4.1}{2.1}{$\Phi$}};
    \draw[->] (2.85,1.02) -- ++(1,0)
                          -- ++(0,-0.4);
    \node at (4.1,0.9) {\tiny flag};
    \node at (4.1,2.5) {\tiny nextstate};
    \node at (4.5,1.85) {\scalebox{0.7}{\tiny flag}};
    \draw[dotted,thick] (-1.5,0.5) -- (-1.5,2);
    \draw[dotted,thick] (1.5,0.5) -- (1.5,2);
    \draw[->, thick] (-1.5,2) -- (1.5,2) node[midway, below] {\tiny Round 1 message};
    \draw[->, thick] (-1.5,1.5) -- (1.5,1.5) node[midway, below] {\tiny Round 2 message};
    \draw[->, thick] (-1.5,1.02) -- (1.5,1.02) node[midway, below] {\tiny Round 3 message};
    \node at (0,0.3) {$\vdots$};
\end{tikzpicture}
    \caption{The architecture of a system and a safety monitor with a privacy-preserving protocol: $(\pi_\Pc,\pi_\Mc)$}
    \label{fig:monitoringsetup}
\end{figure*}
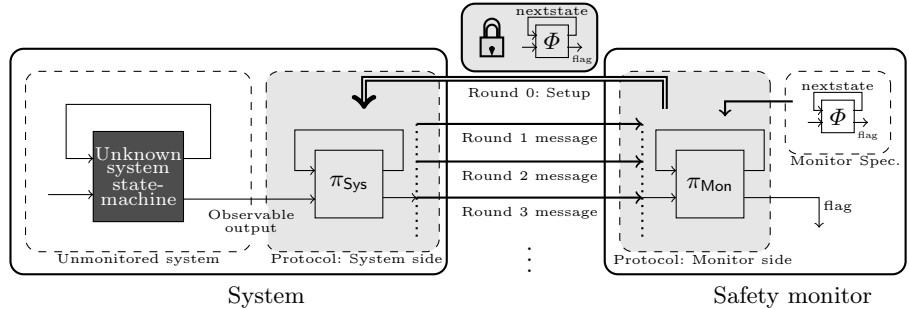


\section{Privacy-preserving monitoring}\label{sec:PrivacyMonitoring}
We describe our setting under which we perform privacy-preserving monitoring. We assume that there are two parties: the System (S) and the Monitor (M). We assume the two parties are \emph{semi-honest}, that is, they follow the protocol, but might examine the transcripts to learn more information, as opposed to malicious parties that deviate from the protocol to obtain information. 

\paragraph*{System's observables.} We assume that the System holds data that is represented as a string that is $s$-bits long. The System's observable sequence is temporal, and therefore, at each round $r$, the data output by the System is represented by $\sigma[r]$, and we use $\sigma$ to represent the $\ell$-length sequence $\sigma[1],\dots,\sigma[\ell]$ of data.

\paragraph{Monitor's input.} The monitor has a monitorable specification $\Phi \subseteq \left( {\{0,1\}^s} \right) ^\omega$. Although other monitoring specifications can be considered~\cite{BFFR18}, we consider the most important ones: safety specifications~\cite{KKLSV02}. 
We assume that the specification $\Phi$ is represented as a state machine, that is, it starts at a state represented by $\mu[1]$, and at every round $r$ where the observable output of the System is $\sigma[r]$, the Monitor's state is updated by a \emph{deterministic} function $\mu[r+1] = \mupdate(\sigma[r],\mu[r])$.
The Monitor has a function $\isbad(\sigma[r],\mu[r])$ that is $1$ if a prefix is not in the specification, and $0$ otherwise.
We assume that such specifications are described as a circuit $\Cc$ that encodes both functions $\mupdate$ and $\isbad$. We call an \emph{initialised circuit} as the pair of circuit and initial monitor state $(\Cc,\mu[1])$. 

\paragraph{Ideal settings.} 
We describe the ideal settings with a trusted third party that our protocols must emulate. We describe two settings: one where the Monitor's specification (and therefore the circuit $\Cc$ representing it) is not private, and another where $\Cc$ is also kept private. 
If the specification is known to both parties, the Monitor hands over $\mu[1]$ to the trusted third party in the first round. If the specification is secret from the System, the Monitor also hands over $\Cc = (\mupdate,\isbad)$. 
At every round $r$, the System hands over observable system output $\sigma[r]$ to this trusted third party. The trusted third party computes and returns $\tau[r] = \isbad (\mu[r],\sigma[r])$ and internally stores  the value $\mu[r+1] = \mupdate(\mu[r],\sigma[r])$.
These two settings are described pictorially in \cref{fig:monitoringOpen,fig:monitoringHiding}. 



\subsubsection{A monitoring protocol.} A monitoring protocol is a pair of instructions $\pi= \tpl{\pi_\Pc,\pi_\Mc}$, one for the System and one for the Monitor, such that at each round~$r$, the Monitor and the System send messages to each other as dictated by the protocol. 
A protocol has a round~$0$, which we refer to as the \emph{setup phase}.
Furthermore, in a fixed round, the only messages sent according to the protocol are from the System to the Monitor. The Monitor  (optionally) responds with ``proceed'' or ``terminate''. 
We add the additional restriction of only one message per round of the monitoring protocol, since the process of monitoring needs to be relatively lightweight and not involve several exchanges within each round. 

We first define the correctness of a monitoring protocol $\pi = \tpl{\pi_\Pc,\pi_\Mc}$.
\begin{definition}[Correctness of monitoring protocols with semi-honest parties]
    A protocol $\pi=\tpl{\pi_\Pc,\pi_\Mc}$ is a \emph{correct monitoring protocol} if for any specification described by an initialised circuit $(\Cc,\mu[1])$, System sequence $\sigma$, and security parameter $n$, the output of the Monitor computed in an execution of the protocol $\pi$ is equal to the output of the Monitor in the ideal setting on the same inputs with high probability, that is, probability $>\tpl{1-\frac{1}{n^d}}$ for any fixed~$d\in\Nb$. 
\end{definition}
\subsubsection*{Secure monitoring protocol.} 
We define security of a monitoring protocol based on a comparison between the real and the ideal setting defined below. 
\paragraph{Real view} The \emph{view of the Monitor on protocol $\pi$}, written as $\view_\Mc^\pi(x_M,\sigma,1^n)$ for an execution of a protocol $\pi$ on inputs $x_M$ of the Monitor and System observable output sequence $\sigma$ is defined as a tuple consisting of 
the Monitor's input,  
the internal random bits that were used,
and the messages $m_1,\dots,m_\ell$ received by the Monitor, during the protocol execution, where $m_j$ is the $j^\text{th}$ message. 

\paragraph{Ideal view.} The \emph{simulated or ideal view of the monitor on protocol $\pi$}, written as $\Sc^\ideal_{\Mc,\pi}(x_M,\sigma,1^n)$, is a transcript generated by a simulator $\Sc_\Mc$ (a probabilistic polynomial-time machine) that has access only to the 
inputs and outputs received by the Monitor in the ideal setting. We drop $\pi$ from the subscript, if the protocol is clear from context.  
The view $\view_\Pc^\pi(x_M,\sigma,1^n)$ and the ideal view $\Sc_\Pc^\ideal$ of the System are defined in a similar way. 
\begin{definition}[Secure monitoring with semi-honest parties]
    A protocol $\pi=(\pi_\Pc,\pi_\Mc)$ is a \emph{secure monitoring protocol without specification hiding} for a  specification represented by a circuit $\Cc$ if there are simulators (probabilistic polynomial time machines) $\Sc_\Mc$ and $\Sc_\Pc$ such that for any initial monitor states $\mu[1]$, System's observable sequences and security parameters $n\in\Nb$, we have
    \begin{align*}
        \set{\Sc_\Mc^\ideal(\mu[1],\sigma,1^n)}_{\mu[1],\sigma} &\indistinguishable \set{\view^\pi_\Mc(\mu[1],\sigma,1^n)}_{\mu[1],\sigma}\\
        \set{\Sc_\Pc^\ideal(\mu[1],\sigma,1^n)}_{\mu[1],\sigma}&\indistinguishable \set{\view^\pi_\Pc(\mu[1],\sigma,1^n)}_{\mu[1],\sigma}.
    \end{align*}
    A protocol $\pi=(\pi_\Pc,\pi_\Mc)$ is a \emph{secure monitoring protocol with specification hiding} if there are simulators (probabilistic polynomial time machines) $\Sc_\Mc$ and $\Sc_\Pc$ such that for all initialised circuits $(\Cc,\mu[1])$, System observable sequence $\sigma$, and security parameters $n\in\Nb$, we have
    \begin{multline*}
        \set{\Sc_\Mc^\ideal((\Cc,\mu[1]),\sigma,1^n)}_{(\Cc,\mu[1]),\sigma}  \indistinguishable\\\set{\view^\pi_\Mc((\Cc,\mu[1]),\sigma,1^n)}_{(\Cc,\mu[1]),\sigma}
    \end{multline*}
    \begin{multline*}
        \set{\Sc_\Pc^\ideal((\Cc,\mu[1]),\sigma,1^n)}_{(\Cc,\mu[1]),\sigma}\indistinguishable\\ \set{\view^\pi_\Pc((\Cc,\mu[1]),\sigma,1^n)}_{(\Cc,\mu[1]),\sigma}.
    \end{multline*}
\end{definition}
    Not that we use the symbol $\indistinguishable$ to denote the standard definition of computational indistinguishability~\cite{BM82}.  
We draw a pictorial representation of the ideal setting of a monitoring protocol. The simulator for each party resides in this is ideal setting, where each party has only the information it receives from the trusted third-party.
\begin{figure}
    \centering
            \begin{tikzpicture}[scale=1, node distance=1cm, >=stealth]

\node[draw, minimum width=4cm, minimum height=4.4cm] (box) at (0,0) {};
\draw[thick,->] (4.5, 2) -- node[above] {$\mu[1]$} (1.1, 2);

\foreach \i/\label in {1/1, 2/2, 3/3, 4/4, 5/5} {
    \draw[->] (-3.3, {2 - 0.6*\i}) -- node[above, sloped] {$\sigma[{\label}]$} (-1.9, {2 - 0.6*\i});
}
\node[text width=1cm] at (-2.1,-1.4)     {$\vdots$};

\foreach \i/\label in {1/2, 2/3, 3/4, 4/5, 5/6} {
    \node at (0, {2.0 - 0.6*\i})  {$\Cc(\sigma[{\i}],\mu[{\i}]) = (\mu[{\label}],\tau[\i])$}; 
}

\foreach \i/\label in {1/1, 2/2, 3/3, 4/4, 5/5} {
    \draw[->] (1.9, {2.0 - 0.6*\i}) -- node[above, sloped] {$\tau[{\label}]$} (3.3, {2.0 - 0.6*\i});
}
\node[text width=1cm] at (3.1,-1.4)     {$\vdots$};

\foreach \i in {2,3,4,5,6} {
    \draw[->] (0.6, {2.4 - 0.6*\i}) to[out=-90,in=90] (-0.3, {2.1 - 0.6*\i});
}
\node[text width=1cm] at (0.4,-1.5)     {$\vdots$};
\node[text width=4cm] at (0.9,-2)     {Trusted third party};
\node[text width=4cm] at (-1.1,-2)     {System};
\node[text width=4cm] at (4.2,-2)     {Monitor};
\end{tikzpicture}

    \caption{Ideal setting with a trusted third party for monitoring where the specification is \emph{not} a secret and circuit $\Cc$ is known to all.}
    \label{fig:monitoringOpen}
\end{figure}
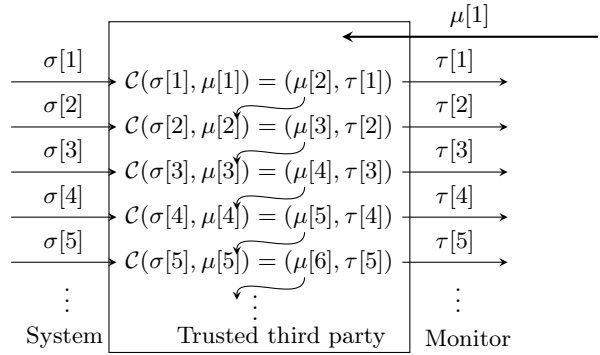
\begin{figure}
    \centering
            \begin{tikzpicture}[scale=1, node distance=1cm, >=stealth]

\node[draw, minimum width=4cm, minimum height=5cm] (box) at (0,0.2) {};
\draw[thick,->] (4.5, 2) -- node[above] {$\mu[1]$} (1.1, 2);
\draw[thick,->] (4.5, 2.5) -- node[above,xshift=6mm] {$\Cc$} (1.1, 2.5);

\foreach \i/\label in {1/1, 2/2, 3/3, 4/4, 5/5} {
    \draw[->] (-3.3, {2 - 0.6*\i}) -- node[above, sloped] {$\sigma[{\label}]$} (-1.9, {2 - 0.6*\i});
}
\node[text width=1cm] at (-2.1,-1.4)     {$\vdots$};

\foreach \i/\label in {1/2, 2/3, 3/4, 4/5, 5/6} {
    \node at (0, {2.0 - 0.6*\i})  {$\Cc(\sigma[{\i}],\mu[{\i}]) = (\mu[{\label}],\tau[\i])$}; 
}

\foreach \i/\label in {1/1, 2/2, 3/3, 4/4, 5/5} {
    \draw[->] (1.9, {2.0 - 0.6*\i}) -- node[above, sloped] {$\tau[{\label}]$} (3.3, {2.0 - 0.6*\i});
}
\node[text width=1cm] at (3.1,-1.4)     {$\vdots$};

\foreach \i in {2,3,4,5,6} {
    \draw[->] (0.6, {2.4 - 0.6*\i}) to[out=-90,in=90] (-0.3, {2.1 - 0.6*\i});
}
\node[text width=1cm] at (0.4,-1.5)     {$\vdots$};
\node[text width=4cm] at (0.9,-2)     {Trusted third party};
\node[text width=4cm] at (-1.1,-2)     {System};
\node[text width=4cm] at (4.2,-2)     {Monitor};
\end{tikzpicture}

    \caption{Ideal setting with a trusted third party for monitoring where the specification is a secret.}
    \label{fig:monitoringHiding}
\end{figure}
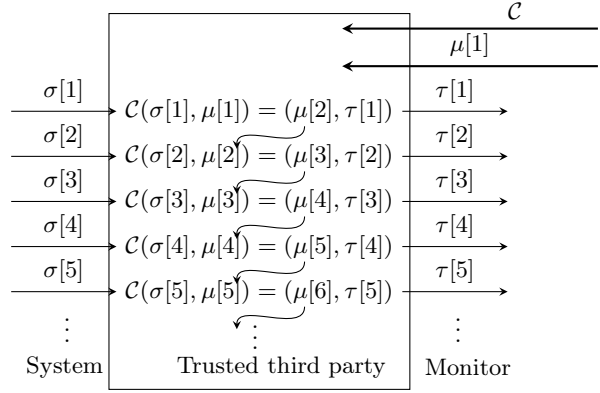

\section{Protocols for privacy-preserving monitoring}\label{sec:protocol}
\subsubsection*{Warm up--Yao's garbling with one gate.}
The most fundamental tool used in secure two-party computation is attributed to Yao and was dubbed Yao's garbling by Goldreich, Micali, and Wigderson~\cite{GMW87}. 
We first describe a toy-version of the problem posed by Yao. 
Consider two parties $A$ and $B$. Can we have a secure protocol where party $A$ has two bits, and $B$ wants to know the output of gate $G$ computing a Boolean function over two bits, where gate $G$ is known to both parties? Yao's garbling produces a simple solution to this as follows. 

Party $A$ starts by randomly generating strings of a fixed length $L^0$, $L^1$, $R^0$, $R^1$, $S^0$ and $S^1$. The strings $L^0$, $L^1$ correspond intuitively to each value $0$ and $1$ taken by the left input wire of the gate, respectively. Similarly $R^0$ and $R^1$ correspond to the values taken by the right input wire, and $S^0$ and $S^1$ to the output wire of the gate taking values $0$ and $1$, respectively. 
After the labelling step, Party $A$ \emph{encrypts} the label of the output of $G$ using keys that correspond to the input. So, if party $B$ had keys that corresponds to input $(0,1)$, then it can only open the output that would represent $G(0,1)$.
More formally, party $A$ calls 
a subroutine $\textbf{\encYao}_G$ that generates the \emph{garbled gate} which consists of four cipher-texts as follows
\begin{multline}
  \textbf{\encYao}_G  \left( [L^0, L^1], [R^0, R^1], [S^0, S^1] \right)
  \coloneqq \\
    \left\{ \Enc_{L^\alpha,R^\beta}  \left(
        S^{G(\alpha,\beta)} \right) \right\}_{\alpha,\beta \in \{0,1\}} \tag{\$}\label{eq:yao}
\end{multline}
and sends it to party $B$, but the encrypted messages are sent in random order. 
That is, if gate $G$ was an \emph{AND} gate, then the garbled gate would be a random order of the elements $\{\Enc_{L^0,R^0}(S^0),\Enc_{L^1,R^0}(S^0)$, $\Enc_{L^0,R^1}(S^0),\Enc_{L^1,R^1}(S^1)\}$. Party $A$ also sends $\seq{S^0,0}$ and $\seq{S^1, 1}$ to indicate to $B$ that $S^0$ corresponds to bit $0$ and $S^1$ to bit 1. If party $A$'s input for the left and the right gate are $\ell,r\in \{0,1\}$, then she also sends the random labels $L^\ell$ and $R^r$.  

Party $B$ then uses $L^\ell$ and $R^r$ as keys to open the four ciphertexts; with very high probability, only one of the four will open, which corresponds to exactly $S^{G(\ell,r)}$. If $S^{G(\ell,r)} = S^0$, he concludes $G(\ell,r) = 0$, and $1$ if $S^{G(\ell,r)} = S^1$. 

Now, consider the same situation with a gate $G$ over two bits, however, one bit of input is known to party $A$ and the other bit is known to party $B$. Can we modify the above protocol to  ensure that party $B$ learns $G(\alpha,\beta)$ without learning $\alpha$, where $\alpha$ is party $A$'s input and $\beta$ is party $B$'s input? 

\subsubsection*{Oblivious Transfer.}
The oblivious transfer functionality ensures that for two parties where one party, say party $A$, has two strings $x_0,x_1\in\{0,1\}^n$ and the other party, that is party $B$ has one bit $\beta\in\{0,1\}$, \texttt{OT} transfers the bit $x_\beta$ to party $B$, without revealing to $A$ the bit $\beta$, or the bit $x_{1-\beta}$ to party $B$. 
The functionality $\texttt{OT}$ is defined as a functionality from $\{0,1\}^{2n}\times \{0,1\}\mapsto \{0,1\}^n$, where $\texttt{OT}((x_0,x_1),\beta) = x_\beta$. Secure protocols for $\texttt{OT}$ have been known since the 1980s, with the earliest forms proposed by Rabin~\cite{Rab81}, and later improvements and alternative protocols by Killian~\cite{Kil88}, and also by Bellare and Micali~\cite{BM89} and several others. 

Using a protocol for the functionality Oblivious transfer ($\texttt{OT}$) as a sub-protocol, we can now modify the previously described protocol for garbling circuits when party $B$ holds bit $\beta$. 
The garbling protocol proceeds as follows. Party $A$ similarly finds labels $L^0,L^1,R^0,R^1,S^0,$ and $S^1$ corresponding to $0$ or $1$ for the wires, and prepares the garbled gates as described in the previous step. After this, party $A$ and party $B$ run a protocol for Oblivious transfer where $A$ has the labels for the input wire corresponding to $B$'s input $\beta$, say $R_0$ and $R_1$, and $B$ has the input bit. At the end of the OT protocol, $B$ would receive $R_\beta$. Later, party $A$ also sends $L_\alpha$. This way, out of the four ciphertexts with the keys $L_\alpha$ and $R_\beta$, party $B$ can only open the one ciphertext that contains the key $S^{G(\alpha,\beta)}$. He matches this string obtained with $S^0$ or $S^1$, both received from $A$. 
\subsubsection*{Yao's garbling with a circuit.}
Consider the same problem, however, instead of just a simple gate $G$, it is a circuit $\Cc$ that party $B$ wants to use to evaluate on party $A$'s input. Then for each wire $W$ in the circuit, party $A$ similarly prepares random keys to represent the value $W^0$ and $W^1$. Whenever an output wire feeds into an input for a gate, party $A$ uses the same random keys for the output and input wire. 
For each gate $G$ in the circuit $\Cc$, party $A$ computes $\textbf{encYao}_G$ using the corresponding gate's inputs and outputs wire labels. Party $A$ further sends the labellings generated for the output wire along with whether they correspond to $0$ or $1$ to $B$ and the input labels of the wires which correspond to her input. 

Party $B$ can use the keys corresponding to the input of party $A$ and unlock the gates to learn the keys to the next gates in a bottom-up manner and work through the circuit until he obtains the keys corresponding to the output wires. Then party $B$ can match the keys with the corresponding values sent by $A$. 

\subsubsection*{Conventions and notations for protocols.}
We describe our protocol for monitoring, which is a reactive functionality.  Any description of a specification is converted into the circuit described below. 
\begin{itemize}
    \item \textit{Circuit size.} Both the System and the Monitor agree that there are $c$ many gates in the circuit $\Cc$ describing the specification. The circuit has $s+m$ many inputs where the first $m$ of the inputs correspond to the Monitor's state and the rest $s$ to the System's input, which are the observable output.  There are $m+1$ output wires. 
    \item  \textit{Circuit structure.} We assume that every gate is a NAND gate with exactly two wires that feed in and one that feeds out. We often use the terms \emph{left} and \emph{right} feed-in wire to differentiate between such wires for a fixed gate. 
    We assume that any wire that is an output of a gate that is also an output wire of the circuit $\Cc$ does not connect back to any feed-in wires. 
\end{itemize}
\paragraph{Wires.} As discussed, there are two kinds of wires: \emph{feed-out} and \emph{feed-in} wires. The idea is that feed-in wires \emph{feed into a gate} and feed-out wires \emph{feed out of a gate}. Every feed-out wire can be connected to zero or more feed-in wires, but every feed-in wire is connected to exactly one feed-out wire; for this matter, we also take every input wire to circuit $\Cc$ to be a feed-out wire. 

\paragraph{Naming the wires.}\, There are $I = 2c$ feed-in wires (two feeding into each gate) and 
$c+s+m$ feed-out wires (one feed-out wire for each gate, and $s+m$ input wires of the circuit).
We use  $\iota_1,\iota_2,\dots,\iota_I$, to represent the feed-in wires and use $\omega_1,\omega_2,\dots,\omega_{c+s+m}$ to represent the feed-out wires. 
Among these feed-out wires, $m + 1$ wires are circuit output wires and the other $O = c+s-1$ feed-out wires are not circuit-output wires. 
The first $m+s$ of the feed-out wires $\omega_1,\dots,\omega_{m+s}$, represent the input wires to the circuit. The first $m$ corresponds to the Monitor's input and the next $s$, the System's input to the protocol.

\paragraph{Gates.}\, We call the $c$ gates $G_1,G_2,\dots,G_c$. The gate $G_j$ has $\iota_{2j-1}$ as 
its left feed-in wire and $\iota_{2j}$ as its right feed-in wire. 
For each $j\in\{1,\dots,c\}$, the wire $\omega_{m+s+j}$ denotes the feed-out wire of gate $G_j$.
The feed-out wires $\omega_{m+s},\dots,\omega_{O+m+1}$ denote the output wires of the circuit. Therefore, the last $m+1$ feed-out wires of gates $G_{c-m-1},\dots,G_{c}$ correspond exactly to the last $m+1$ output wires $\omega_{O+1},\dots,\omega_{O+m+1}$, respectively. The output wires of the monitor state that are used for feed-back into the next round, are represented by $\omega_{O+1},\omega_{O+2},\dots,\omega_{O+m}$ and the output wire corresponding to $\isbad$ is represented by $\omega_{O+m+1}$. See the black text in \cref{fig:exponentsingoingwires} for a pictorial representation of the names of the wires.

\subsection{First protocol - Monitoring without specification hiding}
We first provide a conceptually simpler protocol for when the specification and the circuit $\Cc$ representing it are known to both parties. The internal state computed by the circuit is still kept secret. 
The protocol is similar to Yao's garbling with a circuit described earlier. The main modification here is to reuse the labels for the output wires from one round for the Monitor's state for the input component of the Monitor's state in the next round. This operation is described in line $2$ of the protocol. To obtain the output, the Monitor uses the keys to unlock the circuits from  bottom-up, similar to Yao's protocol.

\begin{algorithm}
\floatname{algorithm}{Open Specification Protocol}
\renewcommand{\thealgorithm}{}
\caption{Secure monitoring without specification hiding}\label{protocol1}
\begin{algorithmic}[1]
\REQUIRE Both parties the Monitor and the System have circuit $\Cc$ as described and agree on a security parameter $n$. Additionally, the Monitor holds an $m$ bit-string $\mstring[1] = \mstring_1[1]\cdot\mstring_2[1]\dots\cdot\mstring_m[1]$ that correspond to the valuation of the initial state, and at each round $r$, the System holds a $s$-bit string $\pstring[r] = \pstring_1[r]\cdot\pstring_2[r]\dots\pstring_s[r]$ that corresponds to the current System state. 
\ENSURE At round $r$, the Monitor learns \emph{only} the last output bit of the circuit $\Cc$ on the string $(\pstring[r], \mstring[r])$, and neither party learns the value of $\mstring[r]$ for $r>1$.
\STATE $S$ : For all out-going wires $\omega_i$ other than the wires corresponding to the input of the Monitor, i.e., $i\in\{m+1,\dots,c+s+m\}$, the System generates two random strings $w^0_i[r]$ and $w^1_i[r]$ corresponding to it. 
\STATE $S$ : For the out-going wire $\omega_i$ that also correspond to the $m$ input wires (for all $i\in \set{1,\dots m}$), the System acts depending on round $r$:
\begin{itemize}
    \item \textbf{for $r = 1$,} the system chooses two string $w^0_i[1]$ and $w^1_i[1]$ uniformly at random, and 
    \item \textbf{for $r>1$,}  
    it reuses the labels of the feedback wires from the previous round i.e, $w^0_i[r] \gets w^0_{O+i}[r-1]$ and $w^1_i[r] \gets w^1_{O+i}[r-1]$. 
\end{itemize}
\STATE $S$:  For each in-going wire $\iota_i$ (for all $i\in \{1,\dots,2c\}$), the System assigns the labels from the output wire $\omega_j$ that is connected to it, i.e., $u^0_i[r] = w^0_j[r]$ and $u^1_i[r] = w^1_j[r]$ if the output wire $\omega_j$ is connected to the feed-in wire $\iota_i$. 
\STATE $S$: The System computes for each gate $G_j$, $\textbf{encGG}_j[r]$ and sends the list $\textbf{encGG}_j[r]$, the labels of the feed-in wires to open the gates for the values corresponding to each bit in $\pstring[r]$, that is, 
  $w^{b_1}_{1}[r],\dots,w^{b_s}_{s}[r]$ where $b_i = \pstring_i[r]$, and both the labels of the output flag bit $w^0_{O+m+1}$ and $w^1_{O+m+1}$, and $\textbf{encGG}_j[r]$ is 
  \[
  \textbf{encGG}_j[r] = \textbf{encYao}_{G_j} \begin{pmatrix}
     \left[u_{2j-1}^0[r], u_{2j-1}^1[r]\right], \\
     \left[u_{2j}^0[r], u_{2j}^1[r]\right], \\
    \left[w_{m+s+j}^0[r],   w_{m+s+j}^1[r]\right]
  \end{pmatrix}
  \]
 \STATE $S,M$ : For $r=1$, the System also uses oblivious transfer to send the labels $w^{c_{1}}_{1}[1],\dots,w^{c_m}_{m}[1]$ where $c_{i}$ is the $i^{\text{th}}$ bit of $\mu[1]$ representing the Monitor input. 
 \STATE $S$ : the Monitor ungarbles the circuit using previously obtained Monitor-state output labels and System's new labels to compute the Monitor-state and flag labels.
 \end{algorithmic}
\end{algorithm} 
The following theorem shows that \protocolone~is correct and secure, assuming that the encryption is secure under the chosen plaintext attack (CPA) (\cref{assumption:CPA}~\cite{Bir11}, in \cref{app:protocol1}).
\begin{restatable}{theorem}{Protocolonesecure}\label{thm:Protocol1secure}
        Assuming the chosen encryption $\textbf{Enc}$ is secure under the CPA model and the oblivious transfer protocol is secure in the presence of semi-honest adversaries, \protocolone~is a \emph{correct and secure} monitoring protocol without specification hiding, when both parties are semi-honest, and the number of rounds is a fixed polynomial in the security parameter. 
\end{restatable}

\subsection{Second protocol - Monitoring with specification hiding}
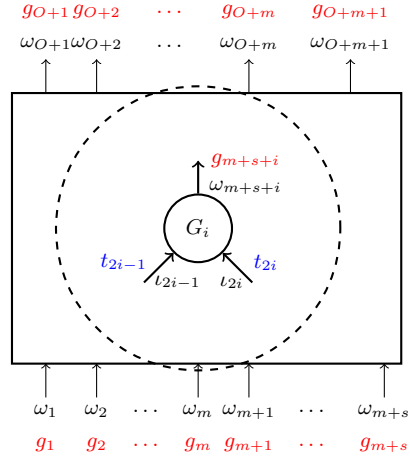
\begin{figure}
\centering
    \begin{tikzpicture}[scale = 0.9, transform shape]
        \draw[thick] (0.25,0) rectangle (6,4);
        
        \node[circle, draw, thick, minimum size=1cm] (A) at (3, 2) {$G_i$};
            
        \draw[->, thick] (A.north) -- (3,3);
        \draw[->, thick] (3.8,1.2) -- (A.south east);
        \draw[->, thick] (2.2,1.2) -- (A.south west);
        \node at (3.7,2.6) {$\omega_{m+s+i}$};
        \node at (2.7,1.2) {$\iota_{2i-1}$};
        \node at (3.5,1.2) {$\iota_{2i}$};
        \node[red] at (3.7,3) {$g_{m+s+i}$};
        \node[blue] at (1.9,1.5) {$t_{2i-1}$};
        \node[blue] at (4,1.5) {$t_{2i}$};

        \foreach \x/\y in {0.75/$\omega_1$, 1.5/$\omega_{2}$, 3/$\omega_{m}$, 3.75/$\omega_{m+1}$, 5.75/$\omega_{m+s}$} {
            \draw[<-] (\x, 0) -- ++(0, -0.5) node[below] {\y};
        }
        \node at (2.25,-0.7) {$\dots$};
        \node at (4.7,-0.7) {$\dots$};
        \foreach \x/\y in {0.75/$g_1$, 1.5/$g_{2}$, 3/$g_{m}$, 3.75/$g_{m+1}$, 5.75/$g_{m+s}$} {
            \node[red] at (\x,-1.2) {\y};
        }
        \node[red] at (2.25,-1.2) {$\dots$};
        \node[red] at (4.7,-1.2) {$\dots$};

        \foreach \x/\y in {0.75/$\omega_{O+1}$, 1.5/$\omega_{O+2}$, 3.75/$\omega_{O+m}$, 5.25/$\omega_{O+m+1}$} {
            \draw[->] (\x, 4) -- ++(0, 0.5) node[above] {\y};
        }    
        \node at (2.6,4.7) {$\dots$};

        \node[red] at (2.6,5.2) {$\dots$};
        \foreach \x/\y in {0.75/$g_{O+1}$, 1.5/$g_{O+2}$, 3.75/$g_{O+m}$, 5.25/$g_{O+m+1}$} {
            \node[red] at (\x, 5.2) {\y};
        }    

        \draw[thick, dashed] (3,2) circle (2.1cm);

    \end{tikzpicture}
    \caption{Base labels: feed-out wires.} \label{fig:exponentsingoingwires}
\end{figure}
\begin{figure}
        \centering
    \begin{tikzpicture}[scale = 0.9, transform shape]
        \draw[thick] (0.25,0) rectangle (6,4);
        
       \node[circle, draw, thick, minimum size=1cm] (C) at (3, 3) {$G_i$};
        \node[circle, draw, thick, minimum size=1cm] (A) at (1.5, 1) {$G_\ell$};
        \node[circle, draw, thick, minimum size=1cm] (B) at (4.5, 1) {$G_r$};
            
        \draw[->, thick] (A.north) -- (1.5,2);
        \draw[->, thick] (B.north) -- (4.5,2);

        \draw[->, thick] (2.3,2.3) -- (C.south west);
        \draw[->, thick] (3.7,2.3) -- (C.south east);
        \draw[->, thick] (C.north) -- (3,3.8);

        \foreach \x/\y in {0.75/$\omega_1$, 1.5/$\omega_{2}$, 3/$\omega_{m}$, 3.75/$\omega_{m+1}$, 5.75/$\omega_{m+s}$} {
            \draw[<-] (\x, 0) -- ++(0, -0.5) node[below] {\y};
        }
        \node at (2.25,-0.7) {$\dots$};
        \node at (4.7,-0.7) {$\dots$};
        \foreach \x/\y in {0.75/$g_1$, 1.5/$g_{2}$, 3/$g_{m}$, 3.75/$g_{m+1}$, 5.75/$g_{m+s}$} {
            \node[red] at (\x,-1.2) {\y};
        }
        \node[red] at (2.25,-1.2) {$\dots$};
        \node[red] at (4.7,-1.2) {$\dots$};

        \foreach \x/\y in {0.75/$\omega_{O+1}$, 1.5/$\omega_{O+2}$, 3.75/$\omega_{O+m}$, 5.25/$\omega_{O+m+1}$} {
            \draw[->] (\x, 4) -- ++(0, 0.5) node[above] {\y};
        }    
        \node at (2.6,4.7) {$\dots$};

        \node[red] at (2.6,5.2) {$\dots$};
        \foreach \x/\y in {0.75/$g_{O+1}$, 1.5/$g_{O+2}$, 3.75/$g_{O+m}$, 5.25/$g_{O+m+1}$} {
            \node[red] at (\x, 5.2) {\y};
        }    

        \node at (2.2,1.7) {$\omega_{m+s+\ell}$};
        \node at (3.8,1.7) {$\omega_{m+s+r}$};

        \node[red] at (0.9,1.7) {$g_{m+s+\ell}$};
        \node[red] at (5.1,1.7) {$g_{m+s+r}$};

        \node[blue] at (1.8,2.4) {$t_{2i-1}$};
        \node[blue] at (4.1,2.4) {$t_{2i}$};

        \node at (2.8,2.3) {$\iota_{2i-1}$};
        \node at (3.4,2.3) {$\iota_{2i}$};

        \node[red] at (1.8,2.9) {$g_{m+s+\ell}^{t_{2i-1}}$};
        \node[red] at (4.1,2.9) {$g_{m+s+r}^{t_{2i}}$};

        \path[->]  (1.5,2) edge [thick,dashed]  (2.3,2.3)
                    (4.5,2) edge  [thick,dashed] (3.7,2.3);
        \draw[thick, dashed] (3,2) circle (2.7cm);

    \end{tikzpicture}
    \caption{Base labels: feed-in wires}\label{fig:labellingingoingwires}
\end{figure}
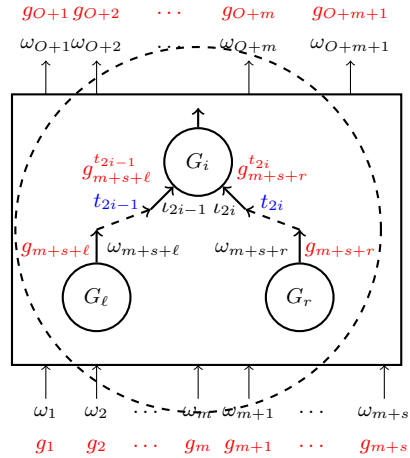
The challenge with designing a protocol where the circuit must also be hidden is that the System cannot come up with the labels for the gates in the circuit, since this requires that the System know the topology of the circuit. 
As a first step, since the topology of the circuit must be kept secret, we assume that the circuit contains only NAND gates.
We provide a protocol, which is a modification of the protocol of Liu et al.~\cite{LWY22}. In our protocol, the Monitor helps the System to come up with the labellings for each gate and also obfuscates the topology of the circuit in this process. The Monitor does so by using a cyclic group $\Gb$ of prime order $q$ where the Decisional Diffie-Hellman (DDH) assumption holds. In a cyclic group, for any non-unitary element $g$, and for two values in $a,b\in \Zb_q$, it holds that $\tpl{g^{a}}^b = \tpl{g^{b}}^a$. 

The Monitor assigns \emph{base labels} to the feed-out wire of each gate $G$ using a randomly chosen element, say $g_G$, from the group $\Gb$. If the feed-out wire of the gate $G$ connects to some feed-in wire, then the monitor randomly chooses exponent $t\in\Zb_q$ for that feed-in wire and labels this wire using $(g_G)^t$.

More specifically, the Monitor selects the base labels using randomly generated group element $g_i$ (represented in red in \cref{fig:exponentsingoingwires}) for each feed-out wire $\omega_i$. 
Then, for each feed-in wire $\iota_j$, the Monitor also chooses an exponent $t_{j}$ (represented in blue in \cref{fig:exponentsingoingwires}).
Finally, the Monitor computes and sends the base labels for each feed-in wire, as follows. For a gate $G_i$ with feed-in wires $\iota_{2i-1}$ and $\iota_{2i}$ that are connected to feed-out wires from gates $G_\ell$ and $G_r$ (as show in \cref{fig:labellingingoingwires}), the monitor uses as labels 
$g_{m+s+\ell}^{t_{2i-1}}$ and $g_{m+s+r}^{t_{2i}}$ as the base labels. 
Note that $g_{m+s+\ell}$ and $g_{m+s+r}$ are the base-labels of the corresponding feed-out wires of gates~$G_\ell$ and~$G_r$.  

The monitor sends three base labels for each gate: two for the feed-in wires and one for the feed-out. Roughly, the DDH assures that ``random exponents of group elements'' cannot be distinguished from ``random group elements''. 
Since DDH assumption holds, we can also show that given $n$ group elements $g_1,g_2,\dots,g_n$ as well as some labels $g_{x_1}^{t_{1}},g_{x_2}^{t_{2}}, \dots, g_{x_n}^{t_{n}}$, the System cannot tell which of the $g_{x_i}^{t_{i}}$s is obtained by exponentiating which of the $g_{i}$, thus successfully obfuscating the circuit topology.

Finally, the System uses these base labels to prepare the labels of the wires that correspond to $0$ and $1$. This is done by choosing one exponent to correspond to the value $0$ and one exponent to correspond to $1$, say $\alpha_0$ and $\alpha_1$. 
For the base label $g$, the System would then label each $\tpl{g}^{\alpha_0}$ and $\tpl{g}^{\alpha_1}$ corresponding to  bits $0$ and $1$, respectively.
Feed-in wire labels are of the form $\tpl{g^t}^{\alpha_i}$, which is equal to $\tpl{g^{\alpha_i}}^t$. So, by knowing the label $g^{\alpha_i}$ (obtained by opening a garbled gate, or from the message of the System), the Monitor can compute $\tpl{g^{\alpha_i}}^t = \tpl{g^t}^{\alpha_i}$. 

Therefore, the Monitor obtains the key to open future gates by simply exponentiating the label obtained from the ungarbling process of a preceding gate and exponentiating with an appropriate exponent~$t$. The System, under the Decisional Diffie-Hellman (DDH) assumption stated below, cannot learn the topology only given such exponentiated group elements.  

We show that our protocol is correct and secure under the DDH assumption over groups of prime order.  Both message sizes and the time taken of our protocol is linear, that is, $\Oh(c+s+m)$, assuming the security parameter is a constant. 

\begin{algorithm}
\floatname{algorithm}{Hidden Specification Protocol}
\renewcommand{\thealgorithm}{}
\caption{Secure monitoring with specification hiding}\label{protocol2}
\begin{algorithmic}[1]
    \REQUIRE Similar to \protocolone, however, only the Monitor knows the specification circuit $\Cc$, but both parties agree on the number of gates $c$ of the circuit, a security parameter $n$, and a group $\Gb$ of order $q\in\Theta(2^n)$.\\
    \ENSURE  At round $r$, the Monitor learns \emph{only} the last output bit of the circuit $\Cc$ at round $r$, and neither party learns the value of $\mstring[r]$ for $r>1$, and the System does not learn $\Cc$.
    \STATE\textbf{Setup.}
  \textsc{Base labels of feed-out wires:} (See \cref{fig:exponentsingoingwires})
\begin{itemize}
    \item[a.]   The Monitor picks random group elements $g_i\in \Gb$ for each feed-out wire $\omega_i$ for $i\in\{1,\dots, O\}$ and sends the list $[g_1,g_2,\dots,g_O]$.     
    \item[b.] Further, the output wires, which correspond to the $m$ feed-out wires $\omega_{O+1},\dots,\omega_{O+m}$ are also given the (same) group elements $g_1,g_2,\dots,g_m$, respectively. The final output wire computing $\isbad$ is represented by $\omega_{O+m+1}$ and is not assigned a group element during setup phase. 
\end{itemize}
\textsc{Base labels of feed-in wires:} (See \cref{fig:exponentsingoingwires,fig:labellingingoingwires})
\begin{itemize}
    \item[c.] The Monitor further picks different exponents $t_i \xleftarrow{\$}\Zb_q$ for each feed-in wire $i\in \{1,\dots,I\}$, and finds the map $\pi: \{1,\dots,I\}\mapsto \{1\dots,O\}$ (chosen uniformly at random) such that $\pi(i) = j$ iff the feed-out wire $\omega_j$ connects to the feed-in wire $\iota_i$.
    \item[d.] The Monitor then computes $\ell_i = g_{\pi(i)}^{t_{i}}$ for every $i\in \{1,\dots,I\}$ and creates list $L = \left[\ell_1,\ell_2,\dots,\ell_I\right]$ and sends this list to the System.     This assigns group element $g_{{\pi(i)}}^{t_i}$ to the feed-in wire. 
\end{itemize}
\STATE \textbf{Labelling wires at round $r\geq 1$} 
   \begin{itemize} 
    \item[e.]  For the first round, $r=1$, the System  picks random, distinct values $\alpha^0[r],\alpha^1[r]\xleftarrow{\$}\Zb_q$. 
    For subsequent rounds, $\alpha^0[r]\gets \beta^0[r-1]$ and $\alpha^1[r]\gets \beta^{1}[r-1]$.
    \item[f.] The System assigns values $w_j^0[r]$ and $w_j^1[r]$, which corresponds to the feed-out wire $\omega_j$ having value 0 and 1, respectively, to the feed-out wires for all $j\in \{1,\dots, O\}$, $w_j^0[r] \gets g_{j}^{\alpha^0[r]} \text{ and }  w_j^1[r] \gets g_j^{\alpha^1[r]}.$
    \newline For all rounds $r$, the System  also selects random values $\beta^0[r],\beta^1[r]\xleftarrow{\$} \Zb_q$ and  it computes the label of the feed-out wires for $j\in \{O+1,\dots,O+m\}$ as $w_j^0[r] = g_{j-O}^{\beta^0[r]} \text{ and }  w_j^1[r] = g_{j-O}^{\beta^1[r]}$. 
    \newline For the feed-out wire $\omega_{O+m+1}$ representing the output of $\isbad$, it assigns 
      $w_{O+m+1}^0[r]\xleftarrow{\$} \Gb \text{ and }  w_{O+m+1}^1[r]\xleftarrow{\$}\Gb.$ \newline The System labels the feed-in wires, for each $i\in \{1,\dots,I\}$, $u_i^0[r] = {\ell_i}^{\alpha^0[r]}\text{ and } u_i^1[r] = {\ell_i}^{\alpha^1[r]}.$
    \item[g.] Proceed as in Steps 4., 5., and 6., in \protocolone, where the System garbles the gates and sends the desired keys, and Monitor ungarbles.
    \end{itemize}
\end{algorithmic}
\end{algorithm}

\begin{assumption}[Decisional Diffie–Hellman assumption \cite{Bon98,Can11}]\label{assumption:DDH}
    In a cyclic group $\Gb = \seq{g}$ of prime order $q\in \Theta(2^k)$ for $n\in \mathrm{poly}(k)$, where $g$ is a generator for $\Gb$, the following probability distributions are computationally indistinguishable:
         $(g^a,g^b,g^{ab})$, where $a$ and $b$ are uniformly and independently at random chosen from $\Zb_q$, and 
         $(g^a,g^b,g^c)$, where $a$, $b$ and $c$ are uniformly and independently at random chosen from $\Zb_q$.
\end{assumption}
\begin{restatable}{theorem}{Protocoltwosecure}\label{thm:Protocol2secure}
    Assuming that the chosen encryption $\textbf{Enc}$ is secure under the CPA model, the DDH assumption on group $\Gb$ holds, and the oblivious transfer protocol is secure in the presence of semi-honest adversaries, \protocoltwo~is a \emph{correct and secure monitoring protocol with specification hiding} when both parties are semi-honest, and the number of rounds is a fixed polynomial in the security parameter. 
\end{restatable}
The proof of correctness of our theorem is a simulation based proof that constructs intermediate indistinguishable transcripts by substituting elements of the transcript with random group elements. A key result required to prove indistinguishability is a corollary of a lemma from the work of Naor and Reingold~\cite{NR04} that is an equivalent representation of the DDH assumption in our proofs, we prove the theorem. A similar lemma is also used in the work of Liu, Wang, and Yu~\cite[Lemma~3]{LWY22}. 

 \begin{restatable}[\protect{\cite[Lemma~4.4]{NR04}}]{lemma}{NaorReingold}\label{lemma:DDHprop}
         Assuming that the DDH assumption holds in a cyclic group $\Gb = \seq{g}$ of prime order $q\in \Theta(2^k)$  for $n\in \mathrm{poly}(\kappa)$, given $n$ randomly chosen elements from the group $g_1,g_2,\dots,g_n\xleftarrow{\$}\Gb$ and $n+1$ randomly chosen exponents $a,a_1,a_2,\dots,$ $a_n\xleftarrow{\$}\Zb_q$,  we have that $\tpl{g_1^{a},g_2^{a},\dots,g_n^{a}}$ is computationally indistinguishable from an $n$-tuple $\tpl{g_1^{a_1},g_2^{a_2},\dots,g_n^{a_n}}$.
 \end{restatable}

\section{Experimental Evaluation}\label{sec:experiments}
To test the feasibility of our protocols for monitoring, we developed an experimental C++ prototype\footnote{This prototype is accessible online at \url{https://github.com/mahykari/ppm}.}
and performed experiments to evaluate the following key questions. 
\begin{enumerate}
    \item How do the measured requirements of both protocols change under varying security parameters ($n\geq 1024$ would be industrial standard),
    \begin{enumerate}
        \item in terms of time taken per round? (see \cref{fig:timekeeper_times_P1,fig:timekeeper_times_P2})\label{item:securityTime}
        \item in terms of message sizes per round? (see \cref{fig:locks_message_P1,fig:locks_message_P2})\label{item:securityMessage}
    \end{enumerate}
    \item When the specification size ($c$) is fixed, but the size of System observables data ($s$) is large, how much do these measurements change for \protocoltwo? (see \cref{fig:timekeeperplus_time_P2,plot:timekeeperplus_size_P2})\label{item:circuitvsprogram}
\end{enumerate}
To answer these questions, we consider the following experiment scenarios:
\begin{enumerate}
\item An \emph{access control system (ACS)} for an office building, where two types of employees, namely types $A$ and $B$, enter or exit the building through a set of external doors. The ACS tracks the numbers of entries and exits for each type of employee and for each door.
The Monitor keeps two variables $\texttt{cnt}_A$ and $\texttt{cnt}_B$,
where $\texttt{cnt}_e$ denotes count of type-$e$ employees currently in the building. At round $r$, for door $i$ of the building, the Monitor receives input from the ACS, structured as follows:
    $\texttt{entered}_A^i[r]$, $\texttt{exited}_A^i[r]$, $\texttt{entered}_B^i[r]$, $\texttt{exited}_B^i[r]$.
  There are $N$ doors, hence $N$ such quadruples. Each $\texttt{entered}_e^i$ denotes the number of type-$e$ employees who have entered through door $i$ of the building \emph{since} the last round; \texttt{exited} values have a similar definition.\\  
  \textit{Specification.} Our specification for this system requires that the number of type $A$ employees currently in the building is never less than type $B$ employees;
  concretely, the \isbad~function of the register machine activates if and only if $\texttt{cnt}_A < \texttt{cnt}_B$.
  Value of $\texttt{cnt}_e$ updates with the following rule:
  $\texttt{cnt}_e \leftarrow
    \texttt{cnt}_e \,+\, \sum_{i} \left(
    \texttt{entered}_e^i[r]\right) - \sum_{i} \left(\texttt{exited}_e^i[r] \right)$.
  Each number is an unsigned integer of fixed bit-width $W$. 
 The monitor only keeps track of employee count per type and needs $2W$ bits for Monitor state, whereas the input of the ACS to each round takes $4NW$ bits.
  
  To answer Question~\ref{item:circuitvsprogram}, we create another case where the ACS keeps track of the internal doors in the building as well (e.g., individual offices). We use $N'$ to denote the number of internal doors. In this case, each ACS update takes $4(N+N')W$ bits, but the specification is still expressed over only $4NW$ bits. 
  \item \emph{The locks of a parallel program}, where every lock has at most one `user' at any given time. Each lock provides a \texttt{lock()} and \texttt{unlock()} interface, and all the locks are initially `unlocked'.
  The Monitor keeps track of all lock states $\texttt{lock}_1, \dots, \texttt{lock}_N$, where $N$ denotes the total number of locks in the system. Each $\texttt{lock}_i$ can have value \texttt{LOCK} or \texttt{UNLOCK}.
  The Monitor, at round $t$, receives input from the lock system, structured as follows: $\texttt{request}_1[t], \dots, \texttt{request}_N[t]$, 
  where $N$ denotes the number of locks, and each \texttt{request} can have value \texttt{LOCK}, \texttt{UNLOCK}, or \texttt{SKIP}.\\
  \textit{Specification.} The specification for this scenario requires that \texttt{lock()} or \texttt{unlock()} is never called twice in a row for any of the locks in the system. Concretely, at round $r$, we have:
  $ \isbad \iff \bigvee_i \left(
    \texttt{request}_i[r] = \texttt{lock}_i \right) $.
  To update $\texttt{lock}_i$, we simply replace its value with $\texttt{request}_i[r]$. 
  As the Monitor keeps track of all locks, it needs $N$ bits for its state.
  Each \texttt{request} takes 2 bits to represent, and hence each update to the Monitor needs $2 N$ bits.
\end{enumerate}
\input{4experimentfiguresarXiv}
\paragraph{Experiment results summary.} 
To answer Question~\ref{item:securityTime}, we plot the breakdown of execution times for \protocolone~in \cref{fig:timekeeper_times_P1} (scenario ACS), and  \cref{fig:timekeeper_times_P2,fig:locks_times_P2} (scenarios ACS and locks) for \protocoltwo. For these cases, ACS only updates on external doors (i.e., $N'=0$),
and we select values
for $N$ among $\{10,30\}$ and $W$ among $\{16,32\}$ to give us $4$ different instances.  
 For the Locks scenario, we consider the values $\{100, 300,500,1000\}$ for the parameter $N$. \protocolone~relies only on random string generation instead of group operations and has symmetric garbling and ungarbling phases and the round times are lower and more uniformly distributed between System and Monitor. 
In the \protocoltwo, time is mostly spent on the System side, since System performs more group operations than Monitor in the garbling process; as the security parameter increases, this difference becomes more visible.
The superlinear growth of round times also conforms with the complexity of group operations. 
\begin{table}
    \centering
    \renewcommand{\arraystretch}{1.1} 
    \setlength{\tabcolsep}{6pt} 
    \begin{tabular}{|l|r|r|}
        \hline
        \multirow{2}{*}{\textbf{Benchmark}} & \multicolumn{2}{c|}{\textbf{Banno et al.}} \\
        \cline{2-3}
        & \textbf{DFA Size} & \textbf{Time} \\
        \hline
        MOD ($m=500$)  & 500  & $\sim$0.002 s \\
        \hline
        BGM ($\psi_2$), REVERSE   & 2885376 & $\sim$24 s \\
        BGM ($\psi_2$), BLOCK     & 11126 & 0.182 s \\
        \hline
        BGM ($\psi_4$), REVERSE  & N/A & time-out \\
        BGM ($\psi_4$), BLOCK     & 7026 & 0.049 s \\
        \hline
        ACS (\cref{fig:timekeeper_times_P2a})  & $\geq 2^{32}$ & $10$ hours (estimated)\\
        \hline 
        LOCKS (\cref{fig:locks_times_P2a}) & $\geq 2^{300}$  & $10^6$ years (estimated) \\
        \hline
    \end{tabular}
    \caption{Monitoring latency of a single event in the protocol of Banno et al.~\cite{BMMBWS22}}
    \label{table:compareBanno}
\end{table}

\begin{table}
    \centering
    \renewcommand{\arraystretch}{1.1} 
    \setlength{\tabcolsep}{6pt} 
    \begin{tabular}{|l|r|r|}
        \hline
        \multirow{2}{*}{\textbf{Benchmark}} & \multicolumn{2}{c|}{\textbf{\protocoltwo}} \\
        \cline{2-3}
        & \textbf{Circuit Size} & \textbf{Time for $n=1024$} \\
        \hline
        MOD ($m=500$)  &  146 & 0.30 s  \\
        \hline
        BGM ($\psi_2$)  & {118} & {0.24 s} \\
        \hline
        BGM ($\psi_4$) & {89} & {0.23 s} \\
        \hline
        ACS (\cref{fig:timekeeper_times_P2a})  & 7127 & 18.94 s \\
        \hline 
        LOCKS (\cref{fig:locks_times_P2a}) & 5700 & 19.96 s \\
        \hline
    \end{tabular}
    \caption{Monitoring latency of a single event in Hidden Specification Protocol}
    \label{table:compareHidden}
\end{table}
For Question~\ref{item:securityMessage}, we plot the breakdown of execution times for \protocolone~in \cref{fig:locks_message_P1} and for \protocoltwo~in \cref{fig:locks_message_P2}. 
For the Locks scenario, we consider the values $\{100, 300,500,1000\}$ again for the parameter $N$ that represents the number of locks.  
As expected, we only observe a linear correlation between message sizes and gate counts in both protocols. 

Finally, for Question~\ref{item:circuitvsprogram}, we plot \cref{fig:timekeeperplus_time_P2,plot:timekeeperplus_size_P2}, which shows how round time increases with increasing sizes of System's observable data, while keeping ``relevant'' System input size fixed and again consider parameters of $N$ and $W$ from $\{10,30\}$ and $\{16,32\}$, respectively. However, to inflate System input size, we use increasing values for the parameter $N'$ (no. of internal doors). Observe that the circuit size remains the same even while the specification considers increasing number of external doors. 
With the parameters we have used in our experiments, as System input increases, the time taken per round increases only marginally, as a result of the garbling phase's dominance in execution time (\cref{plot:timekeeperplus_time_P2}). The increase in message size is more observable, as one key is sent per each bit of the input (\cref{plot:timekeeperplus_size_P2}). However, even this growth is less steep than the growth of the message size when both circuit size and System observables are scaled proportionally as seen in \cref{fig:locks_message_P2}.

We remark that all the monitoring latencies reported are for a single event rather than a trace. 

As discussed in the related work section, the setting considered by Banno et al.~\cite{BMMBWS22} (and therefore also Waga et al~\cite{WMSMBBS24}) is orthogonal to ours---each protocol is tailored to its specific context and not directly applicable to the other. However, some of their specifications they describe might be relevant to our setting, and vice versa. Therefore, we evaluated our protocol on several of their specifications and, conversely, estimated the performance of their protocol on the ACS and locks scenarios from our work.

Banno et al. consider two scenarios: DFA that counts number of 1s in its input modulus $m$ (MOD) and Blood Glucose Monitoring (BGM). 
They also have two protocols REVERSE and BLOCK. 
We implemented the specifications for each scenario with the \emph{highest} reported monitoring latencies in their work for either REVERSE or BLOCK, with the same security parameter as Banno et al. ($n=1024$). For these values, our protocols take time that is in the order of magnitude of 100 milliseconds. We remark that their experiments were run on Intel Xeon Silver (32 cores and 64 threads), a superior hardware to ours. These times are summarised in \cref{table:compareBanno}.

To test our specifications against their protocol, 
we also estimated the time taken by their protocol on the ACS and Locks scenarios we designed. However, our scenarios cannot be directly specified as LTL expressions over the observable output alone to be used directly as an input into their protocol, since our description language is more expressive than LTL specification.  
Since their protocols converts LTL specifications into DFAs,  we considered the size of the smallest possible DFAs accepting the ACS and Locks scenarios to estimate the running time. We then extrapolated the monitoring latency on ACS and the locks scenario from the fact that  their protocol is linear in the size of the DFA, and scale from the other DFA instances provided in their work. 
These times are summarised in \cref{table:compareHidden}. 

For both \cref{table:compareBanno,table:compareHidden}, we use the same notation as their paper to refer to the LTL specifications ($\psi_2$ and $\psi_4$) in their work, obtained originally from related work on runtime verification for artificial pancreas~\cite{CFMS15}. 

The sizes of the formula considered in the experiments conducted by Waga et al.~\cite{WMSMBBS24} are similar in terms of DFA sizes to those considered by Banno et al., and we therefore only restrict our comparison to Banno et al's work. It is reasonable to expect that the monitoring latency of both Waga et al's protocols on these  specifications  would also be in the same order of magnitude. 



\paragraph{Execution pipeline.} 
We write each specification as a synthesisable Verilog module, which describes the \mupdate~and \isbad~functions of our intended register machine. We then synthesise the equivalent circuit, which is what we will use throughout the protocol execution.

We spawn Monitor and System as separate processes that interact via message passing. Internally, we have modelled and implemented each party as a communicating transition system, with access to asynchronous communication channels.
Each state of the transition system can perform read or write operations on a dedicated memory fragment. This design allows us to invoke a protocol within a protocol, as every state can internally keep track of the execution of another transition system, and proceed once the internal protocol is done. An example of this use case is our OT step in the initialisation phase of the \protocoltwo.

We use the GNU Multiple Precision Arithmetic (GMP) library for big-integer arithmetic, OpenSSL for hash functions, and ZeroMQ for asynchronous message passing. We also use the Yosys open synthesis suite for synthesising Boolean circuits from specifications. 
We pass a Verilog specification module to Yosys, with ABC \cite{BM10} as synthesis back-end, to obtain a Boolean circuit equivalent, represented in the Berkeley logic interchange format (BLIF).
Note that all optimisations on circuit size are also done by Yosys and ABC. 
\paragraph{Cryptographic primitives.} The DDH assumption holds in the \emph{quadratic residue} group $QR_q$ if $q$ is a \emph{safe prime} \cite{Bon98}. We use groups $QR_q$ in our protocols,
with values of $q$ defined in RFC 2409, 3526 \cite{rfc2409,rfc3526}. The security parameter specifies the binary representation size of $q$.
Similar to Huang et al.~\cite{HEKM11}, we construct the symmetric encryption scheme for \textbf{EncYao} (see \cref{eq:yao}) using a hash function; precisely, for keys $L$ and $R$ and secret $S$, we have:
$ \Enc_{L,R} ( S ) = 
  \text{SHAKE-256}\left(L \parallel R\right) \oplus 
  \left(S \parallel 1^{100}\right) $,
where $\cdot \parallel \cdot$ denotes string concatenation, and SHAKE-256 is an \emph{extendable-output} hash function from the SHA-3 family; since group elements (hence, their representation) can be arbitrarily large, we needed a hash function with arbitrarily large output.
Note that, contrary to the parametrised security level for the groups we use, SHAKE-256 has a fixed security level of 256 bits; however, this does not impose any practical vulnerabilities and, therefore, is practically secure.
The 100-bit constant padding at the end of the secret is necessary for the decryption phase (gate un-garbling), to detect the correctly decrypted value.
We use the simple OT protocol introduced by Bellare and Micali~\cite{BM89} in both protocols. 

All experiments were run on a personal computer with an Intel Core i5-1235U processor, 16 GB of memory, running Linux Mint 21.3.
Both System and Monitor processes were spawned in parallel, and bound to \texttt{localhost} for communication. We use a timeout of one hour per protocol round, excluding initialisations that take place only in the first round. As a source of true randomness, we periodically read from the file \texttt{/dev/urandom}; we buffer a fixed number of such values in program memory, in order to perform fewer file I/O operations.

\paragraph{Remark.} 
Unlike \protocoltwo, which requires circuits with a single type of binary gate, \protocolone~supports multiple gate types, enabling the use of optimised garbling techniques \cite{ZRE15,App16}. 
Further, even \protocoltwo can be made more efficient if the number of gates of each type is known to all, while still ensuring the circuit topology is not known. 
This opens potential avenues for optimisation. Additionally, our experiments showed a 50\% average reduction in circuit size when ABC utilised all basic gates instead of only NAND gates, which could contribute to further speed-up. 

\section{Outlook}
We took a first step toward privacy-preserving monitoring by proposing protocols that are correct and secure. Our experiments demonstrate an increase in both the message length and the protocol’s overhead with respect to the security parameter, demonstrating a trade-off between privacy and efficiency.

While the levels of privacy provided by our protocols are sufficient for many real-world applications, they fall short of the requirements in highly privacy-sensitive settings.
First of all, we only focus on cryptographic privacy in our work, and further, we only assume semi-honest parties in our current work. Extending these protocols to protect against actively malicious parties---those who intentionally deviate from the protocol---may be computationally expensive. A potential compromise is to consider \emph{covert systems}, as proposed by Aumann and Lindell~\cite{AL10}, where adversaries that deviate from the protocol are caught with a positive probability. For monitoring applications, repeated interactions increase the likelihood of catching cheating agents across multiple rounds, ultimately approaching probabilities close to~1.

Even for the setting of semi-honest parties, we believe that our protocol could be enhanced by optimising the number of gates that represent the specification (see \cref{fig:timekeeperplus_time_P2}). 
Heuristics can be employed to reduce circuit size, but an alternative approach is to relax the specifications---either in terms of soundness or completeness---depending on the specific requirements, to enable encoding with smaller circuits.
Another direction to improve the performance of our protocols is to parallelise the protocols. The process of garbling gates is inherently parallelisable for both parties, particularly for the system. Similarly, the monitor's task of ungarbling can also be parallelised, with the primary bottleneck being the depth of the circuit. Consequently, finding circuits with lower depth, even if it means increasing the number of gates, could enable faster parallelised algorithms.

Our protocol works for specs described by register automata, or using Yosys. It would be future work to integrate it with state-of-the-art monitoring tools such as BeepBeep~\cite{MKH18}, DejaVu~\cite{HPU18}, or \textsc{MonPoly}~\cite{BKZ17}. 
Lastly, our current protocols assume that the monitor and the monitored system are single entities, and that the monitor relies on a linear order of observations. 
Developing privacy-preserving monitoring protocols for scenarios where the system and monitor are distributed is an interesting research challenge.

 \begin{credits}
 \subsubsection{\ackname}This work is a part of project VAMOS that has received funding from the European Research Council (ERC), grant agreement No 101020093.

We thank anonymous reviewers for pointing us to related work~\cite{BMMBWS22} and for their valuable suggestions that improved this paper.
 \end{credits}
%
%
%
\bibliographystyle{alpha}
\bibliography{biblio}

\newcommand{\etalchar}[1]{$^{#1}$}
\begin{thebibliography}{BMM{\etalchar{+}}22}

\bibitem[AL10]{AL10}
Yonatan Aumann and Yehuda Lindell.
\newblock Security against covert adversaries: Efficient protocols for realistic adversaries.
\newblock {\em J. Cryptol.}, 23(2):281--343, 2010.

\bibitem[App16]{App16}
Benny Applebaum.
\newblock Garbling xor gates ``for free'' in the standard model.
\newblock {\em Journal of Cryptology}, 29(3):552--576, Jul 2016.

\bibitem[BBKL19]{BBKL19}
Muhammed~Ali Bing{\"o}l, Osman Bi{\c{c}}er, Mehmet~Sabir Kiraz, and Albert Levi.
\newblock An efficient 2-party private function evaluation protocol based on half gates.
\newblock {\em The Computer Journal}, 62(4):598--613, 2019.

\bibitem[BFFR18]{BFFR18}
Ezio Bartocci, Yli{\`e}s Falcone, Adrian Francalanza, and Giles Reger.
\newblock {\em Introduction to Runtime Verification}, pages 1--33.
\newblock Springer International Publishing, Cham, 2018.

\bibitem[Bir11]{Bir11}
Alex Biryukov.
\newblock {\em Chosen Plaintext Attack}, pages 205--206.
\newblock Springer US, Boston, MA, 2011.

\bibitem[BKH18]{MKH18}
Mohamed~Recem Boussaha, Rapha{\"e}l Khoury, and Sylvain Hall{\'e}.
\newblock Monitoring of security properties using beepbeep.
\newblock In Abdessamad Imine, Jos{\'e}~M. Fernandez, Jean-Yves Marion, Luigi Logrippo, and Joaquin Garcia-Alfaro, editors, {\em Foundations and Practice of Security}, pages 160--169, Cham, 2018. Springer International Publishing.

\bibitem[BKZ17]{BKZ17}
David~A. Basin, Felix Klaedtke, and Eugen Zalinescu.
\newblock The monpoly monitoring tool.
\newblock In {\em RV-CuBES}, 2017.

\bibitem[BM82]{BM82}
Manuel Blum and Silvio Micali.
\newblock How to generate cryptographically strong sequences of pseudo random bits.
\newblock In {\em 23rd Annual Symposium on Foundations of Computer Science {FOCS}}, pages 112--117. {IEEE} Computer Society, 1982.

\bibitem[BM89]{BM89}
Mihir Bellare and Silvio Micali.
\newblock Non-interactive oblivious transfer and applications.
\newblock In Gilles Brassard, editor, {\em Advances in Cryptology - {CRYPTO} '89, 9th Annual International Cryptology Conference}, volume 435 of {\em Lecture Notes in Computer Science}, pages 547--557. Springer, 1989.

\bibitem[BM10]{BM10}
Robert Brayton and Alan Mishchenko.
\newblock Abc: An academic industrial-strength verification tool.
\newblock In Tayssir Touili, Byron Cook, and Paul Jackson, editors, {\em Computer Aided Verification}, pages 24--40, Berlin, Heidelberg, 2010. Springer Berlin Heidelberg.

\bibitem[BMM{\etalchar{+}}22]{BMMBWS22}
Ryotaro Banno, Kotaro Matsuoka, Naoki Matsumoto, Song Bian, Masaki Waga, and Kohei Suenaga.
\newblock Oblivious online monitoring for safety {LTL} specification via {F}ully {H}omomorphic {E}ncryption.
\newblock In {\em Computer Aided Verification - CAV}, pages 447--468. Springer International Publishing, 2022.

\bibitem[Bon98]{Bon98}
Dan Boneh.
\newblock The decision diffie-hellman problem.
\newblock In Joe~P. Buhler, editor, {\em Algorithmic Number Theory}, pages 48--63, Berlin, Heidelberg, 1998. Springer Berlin Heidelberg.

\bibitem[CFMS15]{CFMS15}
Fraser Cameron, Georgios Fainekos, David~M. Maahs, and Sriram Sankaranarayanan.
\newblock Towards a verified artificial pancreas: Challenges and solutions for runtime verification.
\newblock In {\em Runtime Verification}, pages 3--17. Springer International Publishing, 2015.

\bibitem[CH98]{rfc2409}
David Carrel and Dan Harkins.
\newblock {The Internet Key Exchange (IKE)}.
\newblock RFC 2409, November 1998.

\bibitem[Cra17]{Cra17}
S~Michael Crawford.
\newblock Goodhart{\textquoteright}s law: when waiting times became a target, they stopped being a good measure.
\newblock {\em BMJ}, 359, 2017.

\bibitem[CSW20]{CSW20}
Ran Canetti, Pratik Sarkar, and Xiao Wang.
\newblock Blazing fast ot for three-round uc ot extension.
\newblock In {\em Public-Key Cryptography -- {PKC}}, pages 299--327. Springer International Publishing, 2020.

\bibitem[CV11]{Can11}
Ran Canetti and Mayank Varia.
\newblock {\em Decisional Diffie--Hellman Problem}, pages 316--319.
\newblock Springer US, Boston, MA, 2011.

\bibitem[GDPT13]{GDPT13}
Radu Grigore, Dino Distefano, Rasmus~Lerchedahl Petersen, and Nikos Tzevelekos.
\newblock Runtime verification based on register automata.
\newblock In Nir Piterman and Scott~A. Smolka, editors, {\em Tools and Algorithms for the Construction and Analysis of Systems - 19th International Conference, {TACAS} 2013, Held as Part of {ETAPS}}, volume 7795 of {\em Lecture Notes in Computer Science}, pages 260--276. Springer, 2013.

\bibitem[GMW87]{GMW87}
O.~Goldreich, S.~Micali, and A.~Wigderson.
\newblock How to play any mental game.
\newblock In {\em Proceedings of the Nineteenth Annual ACM Symposium on Theory of Computing}, STOC '87, page 218–229, New York, NY, USA, 1987. Association for Computing Machinery.

\bibitem[Gol04]{Gol04}
Oded Goldreich.
\newblock {\em The Foundations of Cryptography - Volume 2: Basic Applications}.
\newblock Cambridge University Press, 2004.

\bibitem[HEKM11]{HEKM11}
Yan Huang, David Evans, Jonathan Katz, and Lior Malka.
\newblock Faster secure two-party computation using garbled circuits.
\newblock In {\em Proceedings of the 20th USENIX Conference on Security}, SEC'11, page~35, USA, 2011. USENIX Association.

\bibitem[HPU18]{HPU18}
Klaus Havelund, Doron Peled, and Dogan Ulus.
\newblock Dejavu: A monitoring tool for first-order temporal logic.
\newblock In {\em 2018 IEEE Workshop on Monitoring and Testing of Cyber-Physical Systems (MT-CPS)}, pages 12--13, 2018.

\bibitem[IF13]{mid2013report}
Mid Staffordshire {NHS} Foundation Trust~Public Inquiry and R.~Francis.
\newblock {\em Report of the Mid Staffordshire {NHS} Foundation Trust Public Inquiry: Executive Summary}.
\newblock HC (Series) (Great Britain. Parliament. House of Commons). Stationery Office, 2013.

\bibitem[JLAP20]{JLAP20}
Samuel Judson, Ning Luo, Timos Antonopoulos, and Ruzica Piskac.
\newblock Privacy preserving {CTL} model checking through oblivious graph algorithms.
\newblock In Jay Ligatti, Xinming Ou, Wouter Lueks, and Paul Syverson, editors, {\em WPES'20: Proceedings of the 19th Workshop on Privacy in the Electronic Society, Virtual Event, USA, November 9, 2020}, pages 101--115. {ACM}, 2020.

\bibitem[JLK{\etalchar{+}}16]{JLKWHSS16}
Yu~Jiang, Han Liu, Hui Kong, Rui Wang, Mohammad Hosseini, Jiaguang Sun, and Lui Sha.
\newblock Use runtime verification to improve the quality of medical care practice.
\newblock In {\em Proceedings of the 38th International Conference on Software Engineering Companion}, ICSE '16, page 112–121, New York, NY, USA, 2016. Association for Computing Machinery.

\bibitem[KAAP25]{KAAP25}
John Kolesar, Shan Ali, Timos Antonopoulos, and Ruzica Piskac.
\newblock Coinductive proofs of regular expression equivalence in zero knowledge, 2025.

\bibitem[Kil88]{Kil88}
Joe Kilian.
\newblock Founding crytpography on oblivious transfer.
\newblock In {\em Proceedings of the Twentieth Annual ACM Symposium on Theory of Computing}, STOC '88, page 20–31, New York, NY, USA, 1988. Association for Computing Machinery.

\bibitem[KK03]{rfc3526}
Mika Kojo and Tero Kivinen.
\newblock {More Modular Exponential (MODP) Diffie-Hellman groups for Internet Key Exchange (IKE)}.
\newblock RFC 3526, May 2003.

\bibitem[KKL{\etalchar{+}}02]{KKLSV02}
Moonjoo Kim, Sampath Kannan, Insup Lee, Oleg Sokolsky, and Mahesh Viswanathan.
\newblock Computational analysis of run-time monitoring: Fundamentals of java-mac1.
\newblock {\em Electronic Notes in Theoretical Computer Science}, 70(4):80--94, 2002.
\newblock RV'02, Runtime Verification 2002 (FLoC Satellite Event).

\bibitem[KM11]{KM11}
Jonathan Katz and Lior Malka.
\newblock Constant-round private function evaluation with linear complexity.
\newblock In {\em Advances in Cryptology - {ASIACRYPT} 2011 - 17th International Conference on the Theory and Application of Cryptology and Information Security}, volume 7073 of {\em Lecture Notes in Computer Science}, pages 556--571. Springer, 2011.

\bibitem[LAH{\etalchar{+}}22]{LAHPTW22}
Ning Luo, Timos Antonopoulos, William~R. Harris, Ruzica Piskac, Eran Tromer, and Xiao Wang.
\newblock Proving {UNSAT} in zero knowledge.
\newblock In {\em Proceedings of the 2022 {ACM} {SIGSAC} Conference on Computer and Communications Security, {CCS}}, pages 2203--2217. {ACM}, 2022.

\bibitem[LJA{\etalchar{+}}22]{LJAPW22}
Ning Luo, Samuel Judson, Timos Antonopoulos, Ruzica Piskac, and Xiao Wang.
\newblock ppsat: Towards two-party private {SAT} solving.
\newblock In Kevin R.~B. Butler and Kurt Thomas, editors, {\em 31st {USENIX} Security Symposium, {USENIX} Security 2022, Boston, MA, USA, August 10-12, 2022}, pages 2983--3000. {USENIX} Association, 2022.

\bibitem[LP09]{LP09}
Yehuda Lindell and Benny Pinkas.
\newblock A proof of security of {Y}ao's protocol for two-party computation.
\newblock {\em J. Cryptol.}, 22(2):161--188, 2009.

\bibitem[LS09]{LS09}
Martin Leucker and Christian Schallhart.
\newblock A brief account of runtime verification.
\newblock {\em The Journal of Logic and Algebraic Programming}, 78(5):293--303, 2009.
\newblock The 1st Workshop on Formal Languages and Analysis of Contract-Oriented Software (FLACOS’07).

\bibitem[LWS{\etalchar{+}}24]{LWSTRP24}
Ning Luo, Chenkai Weng, Jaspal Singh, Gefei Tan, Mariana Raykova, and Ruzica Piskac.
\newblock Privacy-preserving regular expression matching using {TNFA}.
\newblock In {\em Computer Security - {ESORICS} 2024 - 29th European Symposium on Research in Computer Security}, volume 14983 of {\em Lecture Notes in Computer Science}, pages 225--246. Springer, 2024.

\bibitem[LWY22]{LWY22}
Yi~Liu, Qi~Wang, and Siu-Ming Yiu.
\newblock Making private function evaluation safer, faster, and simpler.
\newblock In Goichiro Hanaoka, Junji Shikata, and Yohei Watanabe, editors, {\em Public-Key Cryptography -- PKC 2022}, pages 349--378, Cham, 2022. Springer International Publishing.

\bibitem[Mea14]{Mea14}
Alex Mears.
\newblock Gaming and targets in the {E}nglish {NHS}.
\newblock {\em Universal Journal of Management}, 2:293--301, 09 2014.

\bibitem[MN04]{MN04}
Oded Maler and Dejan Nickovic.
\newblock Monitoring temporal properties of continuous signals.
\newblock In {\em Formal Techniques, Modelling and Analysis of Timed and Fault-Tolerant Systems}, pages 152--166. Springer, 2004.

\bibitem[MS13]{MS13}
Payman Mohassel and Seyed~Saeed Sadeghian.
\newblock How to hide circuits in {MPC} an efficient framework for private function evaluation.
\newblock In {\em Advances in Cryptology - {EUROCRYPT} 2013, 32nd Annual International Conference on the Theory and Applications of Cryptographic Techniques, 2013}, volume 7881 of {\em Lecture Notes in Computer Science}, pages 557--574. Springer, 2013.

\bibitem[NR04]{NR04}
Moni Naor and Omer Reingold.
\newblock Number-theoretic constructions of efficient pseudo-random functions.
\newblock {\em J. ACM}, 51(2):231–262, mar 2004.

\bibitem[Pnu77]{Pnu77}
Amir Pnueli.
\newblock The temporal logic of programs.
\newblock In {\em 18th Annual Symposium on Foundations of Computer Science, Providence, Rhode Island, USA, 31 October - 1 November 1977}, pages 46--57. {IEEE} Computer Society, 1977.

\bibitem[Rab81]{Rab81}
M.~Rabin.
\newblock How to exchange secrets by oblivious transfer.
\newblock Technical report, Technical Report Tech. Memo TR-81, Aiken Computation Laboratory, 1981.

\bibitem[WMS{\etalchar{+}}24]{WMSMBBS24}
Masaki Waga, Kotaro Matsuoka, Takashi Suwa, Naoki Matsumoto, Ryotaro Banno, Song Bian, and Kohei Suenaga.
\newblock Oblivious monitoring for discrete-time {STL} via fully homomorphic encryption.
\newblock In {\em Runtime Verification - 24th International Conference, {RV}}, volume 15191 of {\em Lecture Notes in Computer Science}, pages 59--69. Springer, 2024.

\bibitem[Yao82]{Yao82}
Andrew~Chi{-}Chih Yao.
\newblock Protocols for secure computations (extended abstract).
\newblock In {\em 23rd Annual Symposium on Foundations of Computer Science}, pages 160--164. {IEEE} Computer Society, 1982.

\bibitem[Yao86]{Yao86}
Andrew Chi-Chih Yao.
\newblock How to generate and exchange secrets.
\newblock In {\em 27th Annual Symposium on Foundations of Computer Science ({SFCS} 1986)}, pages 162--167, 1986.

\bibitem[ZRE15]{ZRE15}
Samee Zahur, Mike Rosulek, and David Evans.
\newblock Two halves make a whole.
\newblock In {\em Advances in Cryptology - EUROCRYPT 2015}, pages 220--250, Berlin, Heidelberg, 2015. Springer Berlin Heidelberg.

\end{thebibliography}
\newpage
\appendix
\section{Proof of correctness and security of \protocolone}\label{app:protocol1}
\Protocolonesecure*

The proof that follows is similar to the detailed and comprehensive proof of Yao's protocol~\cite{LP09} and we modify it for our case. Further, we can use sequential composition theorems for semi-honest models~\cite[Theorem~7.3.3]{Gol04} to simplify our proofs, however to make this self-contained, we instead use the fact that there are secure oblivious transfer protocols, which yield a simulator that can produce an indistinguishable transcript for both parties. 
We formally state the CPA assumption here.
\begin{assumption}[Security under chosen plaintext attack (CPA)~\cite{Bir11}]\label{assumption:CPA} 
A \\ private-key encryption scheme $\Enc = (G, E, D)$ is secure under CPA if it has indistinguishable encryptions in the presence of nonuniform adversaries in which the adversary has the ability to choose plaintexts and view their corresponding encryptions.
\end{assumption}

\subsection*{Simulating the view of the Monitor}

 We always assume that we use a sequential secure implementation of the oblivious transfer protocol, 
which means that there are simulators $S^{\texttt{OT}}_{\texttt{Monitor}}$ and $S^{\texttt{OT}}_{\texttt{System}}$, that can simulate the view of the oblivious transfer protocol that are indistinguishable from a real simulation of the System or the Monitor respectively. 

\paragraph{Some intuition.}
Our simulator picks only one random string per wire instead of the two random strings corresponding to $0$ and $1$ as in the protocol. During the garbling process, this one random string is encrypted using the labels of the feed-in wires, along with two other fake labels generated for the feed-in wire.
For example, if the labels were $L^0,L^1$, for the left feed-in wires $R^0,R^1$, for the right feed-in wires and $S^0,S^1$ for the feed-out wires, and the gate being garbled was an \emph{AND} gate, then the garbled gates would be a random ordering of the four elements $$\{\Enc_{L^0,R^0}(S^0),\Enc_{L^1,R^0}(S^0),\Enc_{L^0,R^1}(S^0),\Enc_{L^1,R^1}(S^1)\}.$$ 
However, the simulator first creates only labels $L$, $R$, $S$ for a gate. Later, for each gate it also creates some random labels $L'$ and $R'$  to produce the set of four cipher texts $$\{\Enc_{L,R}(S),\Enc_{L',R}(S),\Enc_{L,R'}(S),\Enc_{L',R'}(S)\}$$ (irrespective of the gate is \emph{AND}, \emph{OR}, or \emph{NAND}) used. This is the crux of the simulator. We describe formally a simulator of the view of the Monitor and prove indistinguishability below.

\subsubsection{Definition of the simulator $\Sc_M$.}
The simulator, at round $r\geq 1$,  computes the following messages as the messages sent by the System during the protocol.
    \begin{enumerate}
    \item For each feed-out wire $\omega_i$ that does not correspond to Monitor input wires, that is, for $i\in\{m+1,\dots,c+s+m\}$, the simulator picks one random string $w_i[r]$.
    \item For feed-out wires $\omega_i$, when that correspond to the input wires, that is, $i\in\{1,\dots,m\}$
    \begin{enumerate}
        \item for round $r=1$, the simulator also picks randomly $w_i$ for feed-out wires that correspond to the Monitor-input. 
        \item for subsequent rounds, where $r>1$, and for the output wires, the simulator instead selects labels from previously selected values, that is, $w_i[r] \gets w_{O+i}[r-1]$.
    \end{enumerate}
    \item For feed-in wire $\iota_i$, that is, $i\in\{1,\dots,2c\}$, the simulator assigns $u_i[r] = w_j[r]$, where output wire $\omega_j$ is connected to input wire $\iota_i$. Further, the simulator also generates $u_i'[r]$, a random string, for each feed-in wire $\iota_i$.
   
    
    \item\label{item:simulatorYaoGG} Finally, the simulator computes the garbled gates for each gate $j\in \{1,\dots,c\}$, that is, 
      \[
  \textbf{encGG}_j = \textbf{encYao} \begin{pmatrix}
     \left[u_{2j-1}[r], u_{2j-1}'[r]\right], \\
     \left[u_{2j}[r], u_{2j}'[r]\right], \\
    \left[w_{m+s+j}[r],   w_{m+s+j}[r]\right]
  \end{pmatrix}
  \]
  and uses 
   $\textbf{encGG}_j$ for all $j\in \{1,\dots,c\}$ as the message from the System to the Monitor. 
   The simulator sends at each step the keys $w_{1},\dots,w_{s}$ as the keys corresponding to the input wires. 
   The simulator also sends the values $w_{O+m+1}[r]$ and a value $w'_{O+m+1}[r]$ randomly generated value such that $w_{O+m+1}[r]$ corresponds to the output wire. That is, if it was asked to $\terminate$~in the ideal setting, then $w_{O+m+1}[r]$ is sent in the place of the label corresponding $0$, and otherwise $w_{O+m+1}[r]$ is sent in the place of the label corresponding to bit~$1$. 
    \item \textbf{For $r=1$}, and for each $i\in \{1,\dots,m\}$, the simulator runs as a subroutine a simulator $S^{\texttt{OT}}_{\texttt{Monitor}}(w_{s+i}^{b_i}[1],b_i)$ for a secure OT-functionality protocol where $b_i$ corresponds to the $i^\text{th}$ bit in the Monitor state $\mu[1]$.
\end{enumerate} 

\subsubsection*{Indistinguishability of Simulator and view of Monitor.}
We first write a short hand to represent the real as well as the simulated views of the Monitor. 

\noindent\textbf{Simulated view of Monitor.}
\begin{itemize}
    \item We will write $\fakeGarble_{\Cc}[r]$ to refer to the garbled gates generated by the simulated as generated in Step~\ref{item:simulatorYaoGG} of the simulator. 
    \item We refer to using $\Bar{M}_i^{\OT}$, the view of the Monitor of the oblivious transfer protocol using simulator $S^{\OT}_\texttt{Monitor}(u^{b_i}_j[1],b_i)$ where $b_i$ is the $i^\text{th}$ bit in the Monitor state $\mu[1]$.  
    \item We write $\fakeKeys[r]$ to represent
    $$\tpl{\Bar{w}_1[r],\dots,\Bar{w}_s[r], \Bar{w}^0_{O+m+1}[r],\Bar{w}^1_{O+m+1}[r]}$$ which are the list of labels for the input wires and the two labels for the distinguished output wire for round $r$,
    where the input wires described by the simulator are $\Bar{w}_i[r]$ for each $i\in \{1,\dots,s\}$ and  the two distinguished output wires using generated by the simulator as $\Bar{w}^0_{O+m+1}[r]$ and $\Bar{w}^1_{O+m+1}[r]$. The overline is added to enable distinction from the real view. 
\end{itemize}
Therefore,
the simulated view is written as $\tpl{\mu[1],\Bar{r}_M, J[1],J[2],\dots,J[\ell]}$ where $\Bar{r}_M$ is the random seed and 
for round $r=1$,
 $$J[1] = \tpl{\fakeGarble_{\Cc}[1],\fakeKeys[r],\Bar{M}^{\OT}_1\dots,\Bar{M}^{\OT}_m}$$ 
 which consists also of messages for the messages received during the oblivious transfer protocol, and later rounds only consists of the garbled circuit and the labels for the feed-out wires where for round $r\geq 2$, we have 
 $$J[r] = \tpl{\fakeGarble_{\Cc}[k], \fakeKeys[r]}$$

\textbf{Real view of Monitor.} Recall that the distribution of the real protocol would consist of the input of the Monitor $\mu[1]$, the input of the random tape $r_M$ and the messages $m_i$ received during the protocol,  the messages $K[i]$, the messages sent during round $i$. \\Therefore, the view of the Monitor is $\tpl{\mu[1],r_M, K[1],K[2],\dots,K[\ell]}$ if there are $\ell$ rounds.

We refer to by $\Garble_{\Cc}[r]$, the set of garbled gates of circuit $C$ as generated during an real execution of the protocol.
Similarly, we write $\keys[r]$ to represent  $$\tpl{{w}^{\pstring_1[1]}_1[1],\dots,w^{\pstring_s[1]}_s[1], w^0_{O+m+1}[1],w_{O+m+1}^1[1]},$$ where  $\pstring[r]$ is the input of the system on the round $r$ and $\pstring_i[1]$ corresponds to the $i^{\text{th}}$ bit of this input. This corresponds to the list of keys corresponding to the labels of the input wires of the System as well as the two output labels for the distinguished output wire. 

For round $r=1$,
 $$K[1] = \tpl{\Garble_{\Cc}[1],\keys[1] ,M^{\OT}_1\dots,M^{\OT}_m}.$$ 
Notice that the first round consists also of messages for the messages received during the oblivious transfer protocol, and later rounds only consists of the garbled circuit and the labels for the feed-out wires. 
Therefore, for round $r\geq 2$, we have $$K[r] = \tpl{\Garble_{\Cc}[k], \keys[r]}$$ where $\encGG_i[r]$ is the encrypted message consisting of four cipher texts generated by garbling the labels for the feed-in and feed-out wires of gate $G_i$.

We wish to show $$\tpl{\mu[1],\Bar{r}_M, J[1],J[2],\dots,J[\ell]} \indistinguishable \tpl{\mu[1],r_M, K[1],K[2],\dots,K[\ell]}.$$
We proceed in steps, where we can show $\fakeGarble_{\Cc}[i]\indistinguishable \Garble_{\Cc}[i]$ for all $i$, and later combined with the indistinguishability of the views in the oblivious transfer protocol, we get our desired result.
\paragraph*{Indistinguishability of garbled circuit and simulated garbled circuit.}
We first show that garbled circuit and the one generated by the simulator are indistinguishable. This is a standard argument of Yao's protocol, and therefore, we only give an overview. 

Since each garbled gate $j$ consists of an encryption of the same key $w_{m+s+j}$ for any pair of keys for all four combinations of the input keys $\left[u_{2j-1}[r], u_{2j-1}'[r]\right]$, and $
     \left[u_{2j}[r], u_{2j}'[r]\right]$ associated with the input wires. 
We fix an round $r$ and therefore only refer to the strings as $\Garble_{\Cc}$ and $\fakeGarble_{\Cc}$ without reference to the round.

Consider an alternate construction of $\fakeGarble_{\Cc}$  that has the same distribution as the original construction of $\fakeGarble_{\Cc}(G)$ generated by the simulator, but is however obtained from an actual instance of $\Garble_{\Cc}$.
Each gate in $C$ has a garbled gate in $\Garble_{\Cc}$ which consists of 4 ciphertexts. Based on the input of both the System and the Monitor, we compute which wire is going to be evaluated using the input keys, and call such wires active. 
Note that labelling such wires active already requires knowing more than just the input of the System as well as the Monitor, but we remark that nevertheless the final garbled circuit we obtain at the end results in the same distribution as the only that is not given input. 

With the knowledge of the input of all parties in the circuit, and by traversing the circuit from the input wires to the output wires, in the topological order, we modify the garbled gates one by one. 
More rigorously, if for a gate $g$, the left feed-in labels are $L^0,L^1$ and the right labels are $R^0,R^1$ and the output labels are $S^0,S^1$, the garbled gate exactly is (a permutation of) the encryption of $\left\{\Enc_{L^\alpha,R^\beta}\left(
        S^{G(\alpha,\beta)} \right) \right\}_{\alpha,\beta \in \{0,1\}} $
For a gate $g$, let $\gamma$ correspond exactly to the output value of the gate. This can be computed from knowing the input and computing the value of the circuit in a bottom-up manner. Therefore, $S^\gamma$ would be the label that would be obtained on decrypting the garbled gate with the label that was marked ``active''.
We replace all 4 encryptions with just one encryption, $$\left\{ \Enc_{L^\alpha,R^\beta}  \left(
        S^{G(\alpha,\beta)} \right) \right\}_{\alpha,\beta \in \{0,1\}} \to \left\{ \Enc_{L^\alpha,R^\beta}  \left(
        S^{\gamma} \right) \right\}_{\alpha,\beta \in \{0,1\}}$$ 

Observe that in the $\fakeGarble_{\Cc}$ provided by the simulator, as well as the one obtained above, each gate contains only one label that is encrypted using four different keys obtained from labels of input wires. Further, this encrypted plain text is a string that is chosen uniformly at random. Therefore, this alternate construction should have an identical distribution to the garbled gates constructed by our simulator, and therefore, we can refer to $\fakeGarble_{\Cc}$ as the one obtained from the garbled gate above.

We now prove that $\Garble_{\Cc}\indistinguishable \fakeGarble_{\Cc}$ under \cref{assumption:CPA} by providing a series of intermediate garbled gates, which we denote by $\Garble^i_{\Cc}$ for $i\in[c]$, where in $\Garble^i_{\Cc}$, the first $i$ garbled gates are replaced with the ``fake'' garbled obtained as in $\fakeGarble_{\Cc}$. Note that since all gates are replaced, $\Garble^c_{\Cc}$ corresponds exactly to $\fakeGarble_{\Cc}$, whereas, we refer to $\Garble^0_{\Cc}$ to also mean $\Garble_{\Cc}$. 

If it was not true that $\Garble_{\Cc}\indistinguishable \fakeGarble_{\Cc}$, then we know that there is some $i$ such that it is also not true that $\Garble^{i}_{\Cc}\indistinguishable\Garble^{i+1}_{\Cc}$.
This means if there is a distinguisher $\Dc$ for  $\Garble_{\Cc}$ and $\fakeGarble_{\Cc}$, then $$\lvert\Pr[\Dc(\Garble_{\Cc}) = 1]-\Pr[\Dc(\fakeGarble_{\Cc})]=1\rvert > 1/\poly(n)$$
therefore, the same $\Dc$  such that for that $i$, we have  $$\lvert\Pr[\Dc(\Garble^i_{\Cc}) = 1]-\Pr[\Dc(\Garble^{i+1}_{\Cc})]=1\rvert > 1/c\cdot\poly(n).$$ 
Using standard arguments, we can obtain a probabilistic polynomial time distinguisher using $\Dc$ above that distinguishes $\Garble_{\Cc}$ and $\fakeGarble_{\Cc}$, to build a distinguisher that can distinguish the double encrypted texts.
Further, the following lemma says that if $\Enc$ was secure under CPA, using \cref{assumption:CPA}, then $\Enc$ must also be secure under chosen double encryptions.
\begin{lemma}[\protect{\cite[Lemma 4]{LP09}}]~\label{lemma:CPAdouble}
    Let $(G, E, D)$  be a private-key encryption scheme that has indistinguishable encryptions under chosen plaintext attacks. Then $(G, E, D)$ is secure under chosen double encryption.
\end{lemma}
Therefore, if there is a distinguisher for the $\Garble^{i}_{\Cc}$ and $\Garble^{i+1}_{\Cc}$, this contradicts the CPA security in \cref{assumption:CPA}.

\paragraph{Indistinguishability with and without using simulators for oblivious transfer.}
Consider the messages $$K[1] = \tpl{\Garble_{\Cc}[1], \keys[1],M^{\OT}_1,\dots,M^{\OT}_m}.$$ If each $M^{\OT}_i$ in the above was replaced with the messages sent by a simulator for the same oblivious transfer protocol, written as $\Bar{M^{\OT}_i}$, we argue that \begin{multline*}
    \tpl{\Garble_{\Cc}[1], \keys[1],M^{\OT}_1\dots,M^{\OT}_m} \\ \indistinguishable  \tpl{\Garble_{\Cc}[1],\keys[1],\Bar{M}^{\OT}_1\dots,\Bar{M}^{\OT}_m}
\end{multline*} 
This is because of the indistinguishability of each $M^{\OT}_i$ from $\Bar{M_i}^\OT$. 
Further, we know that $\Garble_{\Cc}[1]\indistinguishable \fakeGarble_{\Cc}[1]$. Therefore, we have 
\begin{multline*}
     \tpl{\Garble_{\Cc}[1], \keys[C],\Bar{M}^{\OT}_1\dots,\Bar{M}^{\OT}_m} \\\indistinguishable
      \tpl{\fakeGarble_{\Cc}[1], \fakeKeys[1],\Bar{M}^{\OT}_1\dots,\Bar{M}^{\OT}_m}
\end{multline*}
Since each $\keys[1]$ corresponds to randomly chosen strings in both the protocol as well as the simulator, these parts are indistinguishable in both the real view and the simulated view, thus showing $K[1]\indistinguishable J[1]$.

We only remark that our argument for one fixed round routinely extends to multiple rounds.

\subsection*{Simulating the view of the System}
Since there are no messages sent from the Monitor to the System, this simulation is only of the oblivious transfer protocols. We formally define the simulated view. 
\begin{enumerate}
    \item Using its internal random tape, similar to the protocol, the simulator generates the strings  $w_i^0[r]$ and $w_i^1[r]$ for $i\in\{1,\dots,c+m+s\}$ and assigns values for $u_j^0[r]$ and $u_j^1$ for $j\in \{1,\dots,2c\}$.
    \item However, since the only messages sent in the protocol from the Monitor to the System are during Oblivious transfer.  Since we assume \textsc{ObliviousTransfer} protocol is secure under semi-honest adversaries, the simulator for our protocol runs as subroutine the simulator $S^{\texttt{OT}}_{\texttt{System}}(u_j^{0}[1],u_j^{1}[1])$ for the System, where the inputs of the System are $u_j^{0}[1]$, $u_j^{1}[1]$.
    \item After this step, the System in the protocol receives no messages from the System other than $\proceed$ or $\terminate$ at the end of each round.
\end{enumerate}

\subsubsection*{Indistinguishability of Simulator and view of System.}
The indistinguishability of the transcripts follows from the indistinguishability of the oblivious transfer protocol, since the only messages received from the Monitor are during the oblivious transfer protocol.
\section{Proof of correctness and security of \protocoltwo}\label{app:protocol2}
 \Protocoltwosecure*

To show that our protocol is secure, we construct two simulators in the ideal model that have distributions that are computationally indistinguishable from the real model. 

\paragraph{Some intuition.} To prove the above, for garbling, we use the same idea of building ``fake'' garbled gates where all four cipher texts in the garbled gates are obtained by encrypting the same key. However, to ensure the labels are still indistinguishable, we require the simulators generate labels that are indistinguishable. 
For generating ``fake'' labels, we argue that selecting exponents uniformly at random for each gate to create labels during each round  from the group is sufficient. To prove this, we use a technical lemma (\cref{lemma:DDHprop}) which allows us to prove indistinguishability. 

\paragraph{A technical lemma.} We first recall the direct corollary of a lemma from the work of Naor and Reingold that we use as an equivalent representation of the DDH assumption in our proofs. A similar corollary was used in the work of Liu, Wang, and Yu \cite{LWY22} to prove the correctness of their protocols. 
\NaorReingold*

    To prove correctness is straightforward from the protocol. The only time the protocol fails is when a garbled gate is decrypted with the pair of keys and this opens more than one pair of keys. But clearly, this happens with negligible probability. 

    \subsection*{Simulating the view of the Monitor}
    To prove that the protocol is secure, we show that there is a simulator $\Sc$ that produces the Monitor's view of the transcript that is indistinguishable from a real execution of the protocol, assuming the DDH. 

    Recall that formally, we need to show 
    \begin{multline*}
        \set{\Sc_M^\ideal((\Cc,\mu[1]),\sigma,1^n)}_{(\Cc,\mu[1]),\sigma} \indistinguishable \\\set{\view^\pi_M((\Cc,\mu[1]),\sigma,1^n)}_{(\Cc,\mu[1]),\sigma}
    \end{multline*}
        
 
    Consider the following simulator.
    The messages sent by the protocol during the set-up phase are only by the Monitor, and later, the messages sent are only by the System. 
    Thus the simulated view of the Monitor only includes the messages sent after the set-up phase. 
    Instead of selecting two  exponents per round of the protocol which exponentiates each wire, the simulator instead picks an exponent per gate, per round, and uses it to generate one label per gate. The simulator for garbling after proceeds similarly. 
\begin{itemize}
    \item[0.] The Simulator chooses a uniform random tape for the Monitor; 
this defines the values 
$g_1,g_2\dots,g_O\in \Gb$ along with $t_i\in\Zb_q$, which also gives values $L = [\ell_1,\ell_2,\dots,\ell_I]$ where $\ell_i = g^{t_i}_{\pi(i)}$, and $\pi(i)$ is obtained from the circuit $\Cc$ that the Monitor holds as input. 
\end{itemize}
The simulator computes at round $r\geq 1$ the following messages as the messages sent by the System during the protocol.
    \begin{enumerate}
    \item For first round, $r=1$, the simulator  picks at random distinct values $\alpha_i[r]\xleftarrow{\$}\Zb_q$ for each $i\in \{1,\dots,O\}$.
    \\For subsequent rounds, where $r>1$, and for the output wires, the simulator instead selects exponents from previously selected values where $\alpha_i[r]  \gets \beta_i[r-1]$ for $i\in \{1,\dots,m\}$.
    \item For each round, the simulator further selects random values $\beta_i[r]\xleftarrow{\$} \Zb_q$ for $i\in\{1,\dots,m\}$.
    \item The simulator defines values $w_j[r] = g_{j}^{\alpha_j[r]}$ for all $j\in\{1,\dots,O\}$. For $j\in\{O+1,\dots,O+m\}$, the simulator defines 
    $w_j[r] = g_{j-O}^{\beta_{j-O}[r]}$. 
   
    For the flag feed-out wire $\omega_{O+m+1}$, the simulator assigns a group element at random 
      $w_{O+m+1}[r]\xleftarrow{\$}\Gb$. It also generates a random group element $w'_{O+m+1}[r]$ corresponds to the output bit $b[r]$.
    
    The simulator then labels the feed-in wires 
    $u_i[r] = {g_{\pi(i)}}^{t_i\cdot\alpha_{\pi(i)}[r]}$, for each $i\in \{1,\dots,I\}$, 
    and also assigns values to $u_{i}'[r]$ uniformly at random from the group $\Gb$ for each $i\in \{1,\dots,I\}$.
    \item \label{item:simulatorLiuGG} Finally, the simulator computes the garbled gates for each gate $j\in \{1,\dots,c\}$, i.e, 
      \[
  \textbf{encGG}_j = \textbf{encYao} \begin{pmatrix}
     \left[u_{2j-1}[r], u_{2j-1}'[r]\right], \\
     \left[u_{2j}[r], u_{2j}'[r]\right], \\
    \left[w_{m+s+j}[r],   w_{m+s+j}[r]\right]
  \end{pmatrix}
  \]
  and uses 
   $\textbf{encGG}_j$ for all $j\in \{1,\dots,c\}$ as the message from the System. 
   The simulator sends the keys for the $s$ input wires as $u_{1}[1],\dots,u_{s}[1]$ as the keys corresponding to the input. 
   The simulator also sends the values $w_{O+m+1}[r]$ and $w'_{O+m+1}[r]$ randomly generated value such that $w_{O+m+1}[r]$ corresponds to the output bit $b[r]$.
    \item For round $r=1$, Since we assume \textsc{ObliviousTransfer} protocol is secure under semi-honest adversaries, there exists a
    simulator $\Sc^{\texttt{OT}}_{\texttt{Mon}}(u_j^{b_i}[1],b_i)$ for the protocol that uses only the strings  $u_j^{b_i}[1]$ and bit $b_i$, where $b_i$ corresponds to the $i^\text{th}$ bit in the Monitor state $\mu[1]$. The simulator uses this as a subroutine for each $i$ in a sequential manner. 
 %
\end{enumerate} 

We proceed similarly to the earlier protocol and write notation for the view of the Monitor and the protocol.

\noindent\textbf{Simulated view of the Monitor.} The simulated view is  $$\tpl{\Cc,\mu[1], \Bar{r}_M, J[1],\dots,J[\ell]},$$ where $\Cc, \mu[1]$ denote the input of the Monitor, $\Bar{r}_M$ is a random tape of the Monitor, and finally, $J[r]$ is the messages simulated for round $r$.  We  let
 $$J[1] = \tpl{\fakeGarble_{\Cc}[1], \fakeKeys[1],\Bar{M}^{\OT}_1\dots,\Bar{M}^{\OT}_m}$$ 
and for round $r\geq 2$, we have 
$$J[r] = \tpl{\fakeGarble_{\Cc}[k],\fakeKeys},$$ where the terms are described below. 
\begin{itemize}
    \item In round $r$, we denote by $\fakeGarble_{\Cc}[r]$, the garbled circuit generated by the simulator as in Step~\ref{item:simulatorLiuGG} above. 
    \item For round $r=1$, let $\Bar{M}_i^\OT$ represent the view obtained by using the simulator of the oblivious transfer protocol for the $i^{\text{th}}$ oblivious transfer instance. 
    \item We use $\fakeKeys[r]$ to refer to the labels of the input wires corresponding to the System's input along with the unique output wire in round $r$, that is, $$\fakeKeys[r] = \tpl{\Bar{w}_{1}[r],\dots,\Bar{w}_{s}[r], \Bar{w}^0_{O+m+1}[r],\Bar{w}_{O+m+1}^1[r]}.$$ 
\end{itemize}
We end by recalling that simulated view is $\tpl{\Cc,\mu[1], \Bar{r}_M, J[1],\dots,J[\ell]}$.

\noindent\textbf{Real view of the Monitor.}
We say that view of the Monitor based on a correct execution of the protocol is $\tpl{\Cc,\mu[1],r_M, K[1],\dots,K[\ell]}$, 
where $\Cc$ and $\mu[1]$ are the inputs, $r_M$ the random tape
 $K[i]$ denote the messages simulated for round $i$ and therefore 
$$K[1] = \tpl{\Garble_{\Cc}[1], \keys[1],M^{\OT}_1\dots,M^{\OT}_m}$$ 
 where $\pstring_i[r]$ is the input of the system at round $r$
and for round $r\geq 2$, we have $$K[r] = \tpl{\Garble_{\Cc}[r],  \keys[r]}$$ 
where we denote using $\Garble_{\Cc}[r]$, the garbled gates generated by the protocol in round $r$, and $M^{\OT}_i$ denotes the messages sent in a real execution of the oblivious transfer protocol, and $\keys[r]$ to represent the labels of the keys $$\tpl{w^{\pstring_1[r]}_{1}[r],\dots,w^{\pstring_s[r]}_{s}[r], w^0_{O+m+1}[r],w_{O+m+1}^1[r]}$$ sent.
\subsubsection*{Indistinguishability of Simulator and view of Monitor.} We show that the above simulated view is indistinguishable from the real view if the number of rounds is $r = 1$. Then we show that as along as $r$ is a polynomial in the security parameter, and if the simulated view is indistinguishable for $r$ rounds, then it is also indistinguishable for $r+1$ rounds. 

To show that the views are indistinguishable, we first need the following lemma which is an extension from the work of Liu, Wang, and Yu~\cite[Lemma~3]{LWY22} which we extend to suit our case. 
\begin{lemma}\label{prop:extensionDDHpoly}
        For a group $\Gb$ of order $q \in \Theta(2^n)$, and values $f,k \in \poly(n)$, for $k$  elements $g_1,\dots,g_k$ chosen uniformly at random from $\Gb$, and exponents $e_1,\dots,e_{kf}$ and  exponents  $d_1,\dots,d_f$ chosen  independently and uniformly  at random from $[q-1]$, the following are computationally indistinguishable under the DDH assumption (\cref{assumption:DDH}):
    \begin{itemize}
        \item[(i)] $\tpl{\set{g_i}_{i\in\{1,\dots,k\}},\set{g_i^{d_1}}_{i\in\{1,\dots,k\}},\dots,\set{g_i^{d_f}}_{i\in\{1,\dots,k\}}}$
        \item[(ii)] $\tpl{\set{g_i}_{i\in\{1,\dots,k\}},\set{g_i^{e_i}}_{i\in\{1,\dots,kf\}}}$.
    \end{itemize}
\end{lemma}
\begin{proof}
The distribution above is generated by selecting a \emph{block} of $k$ elements $\{g_i\}_{i\in\{1,\dots,k\}}$ from $\Gb$ followed $f$ many blocks where each block is created by either a single exponent per block or different exponents for each group element~$g_i$.

We create hybrid distributions, where for some $j\leq f$, we resultant distribution is obtained by choosing the first $j$ blocks as in distribution (ii) and the later blocks as in distribution (i). More formally, consider distributions $\distribution_j$ for each $j\in\{1,\dots,f\}$ defined as follows, where the first $j$ groups of the initially chosen $\set{g_i}_{k}$ are exponentiated with randomly chosen exponents, whereas the other elements all have one exponent  per block, that is, an exponent $d_t$ is chosen the $t^\text{th}$ block for $t\in\{j+1,\dots,f\}$ to get the tuple  
   \begin{multline*}\Big(\set{g_i}_{i\in\{1,\dots,k\}},\set{g_i^{e_i}}_{i\in\{1,\dots,jk\}},\set{g_i^{d_{j+1}}}_{i\in\{1,\dots,k\}},\\\dots,\set{g_i^{d_f}}_{i\in\{1,\dots,k\}}\Big)
   \end{multline*}

We recall that \cref{lemma:DDHprop} states that 
the following are computationally indistinguishable
\begin{itemize}
    \item[(a)]$\tpl{\tpl{g_1\dots,g_k},\tpl{g_1^{a},g_2^{a},\dots,g_k^{a}}}$, where  $g_i$s are drawn uniformly at random from $\Gb$ and $a$ is drawn uniformly at random from $\Zb_q$
    \item[(b)]$\tpl{\tpl{g_1\dots,g_k},\tpl{g_1^{a_1},g_2^{a_2},\dots,g_k^{a_k}}}$, where $g_i$s are drawn uniformly at random from $\Gb$ and $a_i$s is drawn uniformly at random from $\Zb_q$
\end{itemize}
We show that if distribution $\distribution_j$ and $\distribution_{j+1}$ are distinguishable, then we can obtain a distinguisher for the two distributions (a) and (b) in the restated version of \cref{lemma:DDHprop}.

Suppose there is a distinguisher between $\distribution_j$ and $\distribution_{j+1}$. Given a set of $k$ elements $\tpl{\tpl{g_1\dots,g_k},\tpl{h_1,h_2,\dots,h_k}}$ we can distinguish weather this is obtained as an instance of (a) or (b) described above (contradicting \cref{lemma:DDHprop}) by selecting random exponents $e_i$ for all $i\in\{1,\dots,jk\}$ and then selecting $d_t$ for all $t\in \{j+1,\dots,f\}$ and then sending the following as input to the distinguisher of $\distribution_j$ and $\distribution_{j+1}$: 
\begin{multline*}
    \Big(\set{g_i}_{i\in\{1,\dots,k\}},\set{g_i^{e_i}}_{i\in\{1,(j-1)k\}},\set{h_i^{e_i}}_{i\in\{(j-1)k+1,\dots,jk\}},\\\set{g_i^{d_{j+1}}}_{i\in\{1,\dots,k\}}, \dots,\set{g_i^{d_f}}_{i\in\{1,\dots,k\}}\Big)
\end{multline*}

Observe that the obtained distribution is identical to distribution $\distribution_j$ if the elements $\set{h_i}_{i\in \{1,\dots,k\}}$ is from (a) and is identical to distribution $\distribution_{i+1}$ if the elements were from (b).
\end{proof}
Notice that labels for $\Garble_\Cc$ are obtained from  the list $G$ or $L$ for feed-out and feed-in wires respectively by exponentiating with two randomly chosen exponents. 
Whereas, only one label is required for creating the ``fake'' garbled gate $\fakeGarble_\Cc$ and is obtained by selecting an exponent uniformly at random for each gate. 
The exponents chosen for each round is selected uniformly at random. 

Observe that therefore, if there are $r$ rounds, where $r$ is a polynomial, then the distributions of the keys 
that can be decrypted for each gate is exactly of the form in \cref{prop:extensionDDHpoly} where $f=r$. 
Recall that we call a label active if that label of the wire corresponds to the output evaluated from that gate. 
Therefore, from \cref{prop:extensionDDHpoly}, the distribution of the active labels from a real view are computationally indistinguishable from the simulated labels, since the distribution of active gates. This is because, the sequence of labels of the feed-out wires at round $r$: $\set{w_{j}^{b_j}[r]}_{j\in\{1,\dots,O\}}$, 
where bit  $b_j$ corresponds to the evaluation of the feed-out wire, across all rounds is identically distributed to (i), whereas the labels produced by the simulator across rounds is identically distributed to (ii).


The final output wires $w^0_{O+m+1}[1],w_{O+m+1}^1[1]$ are also chosen uniformly at random by the protocol as well as the simulator and hence indistinguishable from their counterparts generated by the simulator.

With the above \cref{prop:extensionDDHpoly}, this indistinguishability argument extends to the set of all labels across all polynomially many rounds. The proof of the indistinguishability of the garbled gates follow similarly from the earlier proof of indistinguishability of the garbled gates for Protocol~1 as well as Yao's protocol, combined with the fact that the set of active lables is identically distributed to the labels generated by the simulator.

Using this identical distribution of labels as well as a series of intermediate ``hybrid'' garbled gates are created by replacing the garbled gates from the $\Garble$ with $\fakeGarble$, we can prove indistinguishability in a routine manner. The indistinguishability of the intermediate steps follow due to the encryption scheme $\Enc$ being secure under \cref{assumption:CPA} along with \cref{prop:extensionDDHpoly}.
Which therefore ensures that transcripts created are also indistinguishable subject to \cref{assumption:CPA,assumption:DDH} (CPA, DDH).




    \subsection*{Simulating the view of the System}
    We now show a simulator that can generate the System's view of transcripts that is computationally indistinguishable (under the DDH assumption) from an execution of the protocol. Observe that the only view of the System is the initial set-up phase after which it receives no more messages from the Monitor, other than the oblivious transfer protocol for round $r=1$.

    Consider the following simulator that generates the view for the System. 
    
\begin{enumerate}
    \item  The simulator picks random group elements $g_i$ for each feed-out wire $\omega_i\in O$ and sends the list $[g_1,g_2,\dots,g_I]$ and also sends a list $L$ of labels by selecting elements uniformly at random from the group $\Gc$ to create $L = [\ell_1,\ell_2,\dots,\ell_I]$, where $\ell_i\xleftarrow{\$} \Gb$.
    \item Since we assume \textsc{ObliviousTransfer} protocol is secure under semi-honest adversaries, and in the hybrid model, the simulator runs as subroutine the simulator $\Sc^{\texttt{OT}}_{\texttt{System}}(u_j^{0}[1],u_j^{1}[1])$ for the System, where the inputs of the System are $u_j^{0}[1]$, $u_j^{1}[1]$ for the oblivious transfer protocol (for each $j\in \{1,\dots,m\}$).
    \item After this step, the simulator receives no messages from the System other than $\proceed$ or $\terminate$.
\end{enumerate} 
Since the simulator's only view consists of the messages sent during the setup phase followed by the oblivious transfer protocol in round $1$, it is enough to show that the transcript until round $1$ is indistinguishable. Further, since there are no messages sent by the Monitor after the set-up phase, we do not write the input of the program as a part of the view of the program. 

\noindent\textbf{Simulated view of the System}
   \begin{itemize}
    \item Let the messages sent during the setup-phase be denoted by $\Bar{A}^\setup = \tpl{\Bar{G},\Bar{L}}$, where $\Bar{G}$ is $O$ many randomly generated elements from the group, and $\Bar{L}$ is also $2c$-many randomly generated elements from the group. 
    \item Let $\Bar{P}^\OT_i$ denote the simulation of an execution of the oblivious transfer protocol of the System where the two strings used as by the Monitor for the $j^{\text{th}}$ input in the protocol $\OT$ are $u^0_j[1]$ and $u_j^1[1]$.
   \end{itemize}
Therefore, the view output by the simulation is $$\tpl{\Bar{r}_P,\Bar{A}^\setup, \Bar{P}^\OT_1,\dots,\Bar{P}^\OT_m}.$$

\noindent\textbf{Real view of the program}
   \begin{itemize}
    \item Let the messages sent during the setup-phase be denoted by $A^\setup = \tpl{G,L}$, where $G = [g_1,\dots,g_O]$ is $O$ many randomly generated elements from the group, and $L = [\ell_1,\dots,\ell_{2c}]$ where $\ell_i = g_{\pi(i)}^{t_i}$ and $t_i$ is a random value in $\Zb_q$. 
    \item Let $P^\OT_i$ denote the messages sent by the Monitor to the System  during a real execution of the oblivious transfer protocol where the two strings used as by the Monitor for the $j^{\text{th}}$ input in the protocol $\OT$ are $u^0_j[1]$ and $u_j^1[1]$.
   \end{itemize}
   The real view obtained from an execution of a protocol is  $$\tpl{r_P,{A}^\setup, {P}^\OT_1,\dots,{P}^\OT_m}.$$
\subsubsection*{Indistinguishability of Simulator and view of System.} 

Observe that both in the real execution and the simulated execution, the elements of $G$ are just random $O$ elements, and therefore, both $G$ and $\Bar{G}$ are indistinguishable. For convenience we will use $g_i$  to denote the $i^{\text{th}}$ element of both $G$ and $\Bar{G}$. 
First, we show that $\Bar{L}$, which consists of $2I$ random elements is indistinguishable from $L$ which consists of elements $\ell_i$s such that $\ell_i = g_{\pi(i)}^{t_i}$, where $t_i$ is chosen uniformly at random, and so is the permutation $\pi$.
\begin{lemma}\label{prop:groupProposition}
    The distribution of $I$ random group elements is indistinguishable from the list $L = [\ell_1,\ell_2,\dots\ell_I]$ where $\ell_i = g_{\pi(i)}^{t_i}$, where $\pi \colon \{1,\dots,I\}\to \{1,\dots,O\}$, and $g_1,g_2,\dots,g_O$ are chosen uniformly at random from $\Gb_q$.
\end{lemma}
\begin{proof}
    Consider the following $O$ intermediate distributions to generate a list of length $I$, which results in lists $L_1,L_2,\dots,L_O$. Let $L_0 = L$ obtained as described in the statement of the lemma. 
    We say list $L_i$ is obtained from list $L_{i-1}$ by replacing all instances of elements that are obtained as an exponent of $g_i$, that is $i = \pi(j)$ and $\ell_j  = \tpl{g_i}^{t_j}$, then $\ell_j$ is replaced with an element selected uniformly at random from the group.
    
    By definition $L$ is identically distributed to $L_0$ and we know that $L_O$ is identically distributed to $\Bar{L}$.
    We will show that for any $i$, $L_i \indistinguishable L_{i+1}$, to complete our proof of the above. 
    
    Suppose $O(j) = i$ for $j= j_1,j_2,\dots,j_k$. Since we obtain $L_{i+1}$ by replacing since we are only replacing elements of the form $\tpl{g_i}^{t_{j_1}}, \tpl{g_i}^{t_{j_2}},\dots,\tpl{g_i}^{t_{j_k}}$ with $k$ group elements chosen uniformly and independently at random $g^{a_1}, g^{a_2},\dots, g^{a_k}$ where $g$ is a generator of the group. Since $g_i = g^c$ for some value $c$ chosen uniformly at random, 
    we have that the sequence $\tpl{g_i}^{t_{j_1}}, \tpl{g_i}^{t_{j_2}},\dots,\tpl{g_i}^{t_{j_k}}$ $= \tpl{g^{t_{j_1}}}^c, \tpl{g^{t_{j_2}}}^c,\dots,\tpl{g^{t_{j_k}}}^c$. Therefore, it follows from \cref{lemma:DDHprop} that under the DDH assumption (\ref{assumption:DDH}) that these values are indistinguishable. 
    \end{proof}

    It is an easy extension of the above proof that the joint distribution $(\Bar{G},\Bar{L}$ is indistinguishable from $(G,L)$
Since $\Bar{A}^\setup = (\Bar{G},\Bar{L})$, and $A^\setup = (G,L)$ we have shown that $\Bar{A}^\setup \indistinguishable A^\setup$

Since both the  simulation of $n$ oblivious transfers, and the distribution of the transcript of real executions of $n$ $\OT$ protocols are indistinguishable, we can also further conclude that this joint distribution $\tpl{\Bar{A}^\setup, \Bar{P}^\OT_1,\dots,\Bar{P}^\OT_m}$ $\indistinguishable$ $\tpl{{A}^\setup, {P}^\OT_1,\dots,{P}^\OT_m}$. This shows that the view of System is indistinguishable from the simulated view. 

%
\end{document}